\documentclass[12pt, draftclsnofoot, onecolumn]{IEEEtran}
\IEEEoverridecommandlockouts
\usepackage{amsfonts}
\usepackage{amsmath}
\usepackage{cite}
\usepackage{slashbox}
\usepackage[]{algorithm2e}
\usepackage{subfigure}
\PassOptionsToPackage{hyphens}{url}
\usepackage{url}

\usepackage[pdftex]{graphicx}
\usepackage{epstopdf}
\usepackage{cite}
\usepackage{multirow}
\usepackage{pifont}
\usepackage{capt-of}
\newtheorem{lemma}{Lemma}
\newtheorem{proof}{Proof}
\newcommand{\pb}{p_{\mathsf{b}}}
\newcommand{\pt}{p_{\mathsf{t}}}
\newcommand{\thetat}{\theta_{\mathsf{t}}}
\newcommand{\thetar}{\theta_{\mathsf{r}}}

\newcommand{\Gt}{G_{\mathsf{t}}}
\newcommand{\gt}{g_{\mathsf{t}}}
\newcommand{\Gr}{G_{\mathsf{r}}}
\newcommand{\gr}{g_{\mathsf{r}}}
\newcommand{\alphaL}{\alpha_{\mathsf{L}}}
\newcommand{\alphaN}{\alpha_{\mathsf{N}}}
\newcommand{\mN}{m_{\mathsf{N}}}
\newcommand{\mL}{m_{\mathsf{L}}}
\newcommand{\Nt}{N_{\mathsf{t}}}
\newcommand{\Nr}{N_{\mathsf{r}}}
\newcommand{\rout}{r_{\mathsf{out}}}
\newcommand{\rin}{r_{\mathsf{in}}}
\newcommand{\RB}{R_{\mathsf{B}}}
\newcommand{\pM}{p_{\mathsf{M}}}

\newcommand{\be}{\begin{eqnarray}}
\newcommand{\ee}{\end{eqnarray}}

\usepackage[usenames]{color}

\makeatletter
\def\blfootnote{\xdef\@thefnmark{}\@footnotetext}
\makeatother

\begin{document}
\hyphenation{multi-symbol}
\title{Device-to-Device Millimeter Wave Communications: Interference, Coverage, Rate, and Finite Topologies}
\author{ Kiran Venugopal, Matthew C. Valenti, and Robert W. Heath, Jr. \\ 
\thanks{Preliminary results related to this paper were presented at the 2015 Information Theory and Applications (ITA) Workshop \cite{mmWave:2015}. This work was supported in part by the Intel 5G program and the National Science Foundation under Grant No. NSF-CCF-1319556.  M.C. Valenti was supported by the Big-XII Faculty Fellowship program. Kiran Venugopal and Robert W. Heath, Jr. are with the University of Texas, Austin, TX, USA. Matthew C. Valenti is with West Virginia University, Morgantown, WV, USA. Email: \tt{kiranv@utexas.edu, valenti@ieee.org, rheath@utexas.edu}}
}
%\IEEEauthorblockA{\IEEEauthorrefmark{1}The University of Texas, Austin, TX, USA.\\
%\IEEEauthorrefmark{2} West Virginia University, Morgantown, WV, USA. \\
%\IEEEauthorrefmark{1}\{, rheath\, \IEEEauthorrefmark{2}
% }%}
\date{}
\maketitle

\vspace{-1.5cm}
\thispagestyle{empty}
\begin{abstract}
Emerging applications involving device-to-device communication among wearable electronics require Gbps throughput, which can be achieved by utilizing millimeter wave (mmWave) frequency bands. When many such communicating devices are indoors in close proximity, like in a train car or airplane cabin, interference can be a serious impairment. This paper uses stochastic geometry to analyze the performance of mmWave networks with a finite number of interferers in a finite network region. Prior work considered either lower carrier frequencies with different antenna and channel assumptions, or a network with an infinite spatial extent. In this paper, human users not only carry potentially interfering devices, but also act to block interfering signals. Using a sequence of simplifying assumptions, accurate expressions for coverage and rate are developed that capture the effects of key antenna characteristics like directivity and gain, and are a function of the finite area and number of users. The assumptions are validated through a combination of analysis and simulation. The main conclusions are that mmWave frequencies can provide Gbps throughput even with omni-directional transceiver antennas, and larger, more directive antenna arrays give better system performance.  
\end{abstract}

%%%%%%%%%%%%%%%%%%%%%%%%%%%%%%%%%
\section{Introduction}
\label{sec:Intro}
Wearable devices are positioned to become part of everyday life, whether be it in the realm of healthcare, the workplace, or infotainment \cite{wearable_insight, IDC_forecast}. Mobile wearables open up unique challenges in terms of power consumption, heat dissipation, and networking \cite{Starner2014}.  From a wireless communications perspective, wearable communication networks are the next frontier for device-to-device (D2D) communication \cite{Juniper}. Wearable networks connect different devices in and around the human body including low-rate devices like pedometers and high-rate devices like augmented- or mixed-reality glasses. With the availability of newer commercial products, it seems feasible that many people will soon have multiple wearable devices \cite{Winter:2015}, as illustrated in Fig. \ref{fig:wearble_network}. Such a wearable network \textit{around} an individual may need to operate effectively in the presence of interference from other users' wearable networks. This is problematic for applications that require Gbps throughput like virtual reality or augmented displays. The urban train car will be a particularly bad environment with a high density of independent wearable networks located in close proximity \cite{trainDensity} as illustrated in Fig. \ref{fig:train_car}. Understanding the interference environment is critical to understanding the achievable rate and quality-of-experience that can be supported by wearable communication networks as well as the feasible density of such networks. 

\begin{figure}
\centering
\subfigure[Example wearable communication network. The  user's smartphone can act as a coordinating hub for the wearable network.]{
\includegraphics[scale=0.25]{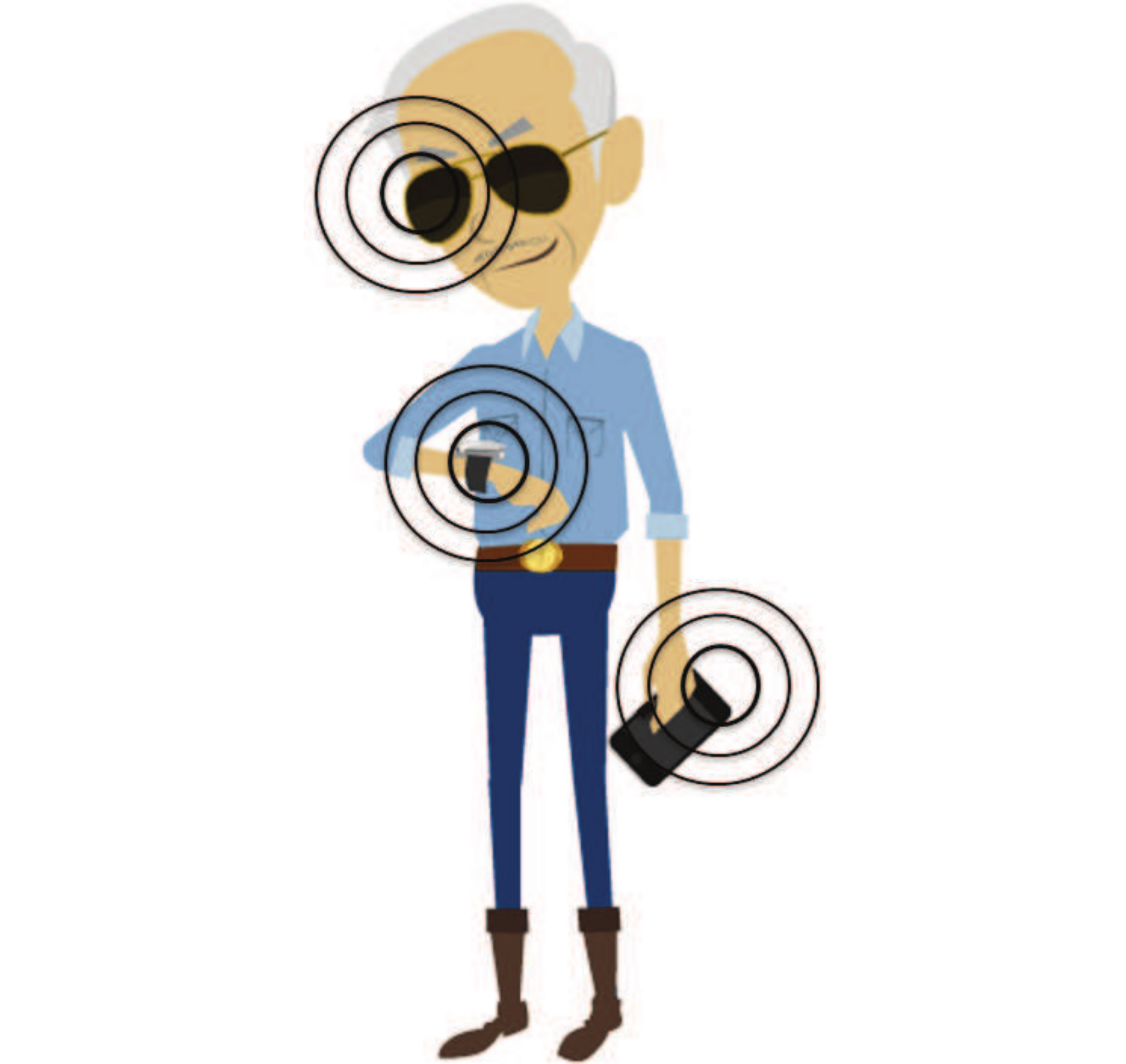}
\label{fig:wearble_network}
}
\hspace{.5in}
\subfigure[A finite network located, for instance, in a train car. Small circles represent wireless devices and large circles represent blockages. Also shown the concept of a blocking cone ${C_i}$, and a blocked interferer $X_j \in {C_i}$.]{
\includegraphics[scale=0.6]{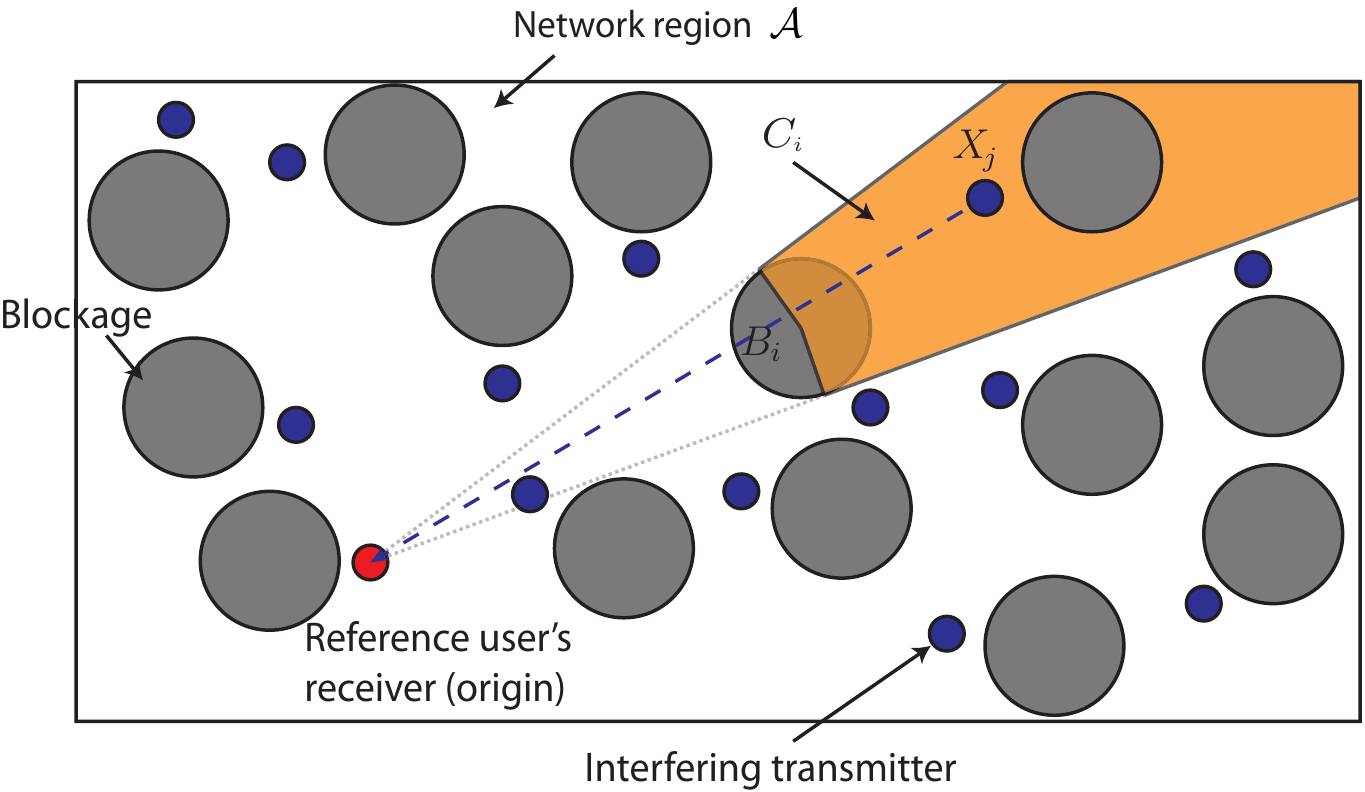}
\label{fig:train_car}
}
\caption{Many users with wearable networks like those shown in (a) will be located in close proximity as in (b), creating mutual interference. People block some of the interfering signals.}
\label{fig:system_model}
\vspace{-.2in}
\end{figure}

The millimeter wave (mmWave) band contains a wide range of carrier frequencies capable of supporting short range high-rate wireless connectivity \cite{CPark2007}. The mmWave band has several desirable features which include large bandwidth, compatibility with directional transmissions, reasonable isolation, and dense deployability. Standards like  Wireless HD \cite{WirelessHD} and IEEE 802.11ad \cite{IEEE:11ad} have already made mmWave-based commercial products a reality. Wearable networks might use these standards or might use device-to-device operating modes proposed for mmWave-based next-generation (5G) cellular systems \cite{mmWCM:2014, Rapp:2013}. Short-range mmWave communication systems usually focus on high-speed wireless connectivity to replace cable connections. However, these emerging protocols have yet to prove their effectiveness in a highly dense interference scenario.

The tool of stochastic geometry has been extensively used to study interference in large wireless networks \cite{cardieri:2010, andrews:2010, andrews:2009, Haenggi:Now}. Prior work on mmWave-based networks has also used the results from stochastic geometry to analyze coverage and rate  \cite{bai:2014, SSingh_arX, Bai:CM2014} while modeling the directionality of antennas and the effect of blockages. For analytical tractability, most work assumes an infinite number of mobile devices spread over an infinite area. These assumptions allow the analytical expressions related to the spatial average of the system performance to be simplified through application of Campbell's theorem \cite{baccelli:2010}. Analysis of the outage probability conditioned on the network geometry in ad hoc networks with a finite spatial extent and number of interferers was performed in \cite{torrieri:2012}, which was extended to the analysis of frequency-hopping networks in \cite{valenti:2013a}. The unique channel characteristics and antenna features \cite{mmWaveBook} for mmWave networks, however, were not considered in \cite{torrieri:2012, valenti:2013a}. The mmWave channel has been studied for the outdoor environment \cite{Rappaport:measurement} and the significant effect of blockages on signal propagation is well known \cite{Bai2014}. In crowded environments such as train cars or airline cabins, human bodies are a main and significant source of blockage of mmWave frequencies \cite{Lu:ZTE, Bai:Asilomar14}. This implies that the very same users that wear the interfering transmitters act to block interference from other wearable networks.

In this paper, we characterize the performance of mmWave wearable communication networks. We focus on networks operating at mmWave carrier frequencies that are confined to a limited region and contain a finite number of interferers while not explicitly modeling the impact of reflections within the finite region or at its boundaries. We develop an approach for calculating coverage and rate in such a network.  As mmWave systems are likely to use compact antenna arrays, we assess the impact of antenna parameters, in particular the beamwidth and antenna gain, on the coverage and spectral efficiency of the system. Compared with \cite{bai:2014, SSingh_arX, Bai:CM2014, Haenggi:Now, torrieri:2012, valenti:2013a}, we use the same computational approach as in \cite{torrieri:2012, valenti:2013a}, with assumptions on mmWave propagation, antennas, and blockage similar to those in \cite{bai:2014, Bai:CM2014}, though we model people -- not buildings -- as blockages. Compared with our prior work in \cite{mmWave:2015} where the interferers were assumed to be at fixed locations, this paper also considers interferers that are randomly located. We begin by presenting an analysis that leads to closed-form expressions for the coverage probability conditioned on the location of the interferers and blockages.  Then, through a sequence of assumptions, we find the spatially averaged coverage and rate when the interference and blockages are drawn from a random point process.  The assumptions and analysis are verified through a set of simulations, which involve the repeated random placement of the users according to the modeled point process. 

The organization of the paper is as follows: We introduce the network topology and signal model in Section \ref{sec:Model}. We describe the interference model and derive expressions for the signal to interference plus noise ratio (SINR) distribution and rate coverage probability in Section \ref{sec:Interfere_model}. In Section \ref{sec:Num_results}, we provide numerical results when the users are located at fixed locations. We assert the simplifying assumptions for analyzing wearable networks when the users are located at random locations in Section \ref{sec:SpatialAvg} and, in Section \ref{sec:sim_results}, verify through simulations that the assumptions have a negligible effect on the accuracy of the analysis. Finally, we conclude our work and give suggestions for future work in Section \ref{sec:Concl}. 

%%%%%%%%%%%%%%%%%%%%%%%%%%%%%%%%%
\section{Network Topology and Signal Model}
\label{sec:Model}
Consider a finite network region $\mathcal A$ with a reference receiver and $K$ potentially interfering transmitters. The reference transmitter is assumed to be located at an arbitrary but fixed distance $R_0$ from the reference receiver at an azimuth angle $\phi_0$ and elevation $\psi_0$. The area of the network in the horizontal plane is denoted by $|\mathcal A|$, so that the interferer density $\lambda = K/|\mathcal A|$. %$\mathcal A$ can be any arbitrary shape (though a circular or annular region is easier analytically).
The interfering transmitters and their locations are denoted by $X_i,~i=1,~2,~...,~K$. We assume the reference receiver to be located at the origin and represent $X_i$ as a complex number $X_i = R_i e^{j \phi_i}$, where $R_i = |X_i|$ is the distance between the $i^{th}$ transmitter and the receiver, and $\phi_i = \angle X_i$ is the azimuth angle to $X_i$ from the reference receiver. For simplicity, we assume that all the interferers are on the same horizontal plane that contains the reference receiver, though our model could also be easily generalized to handle the 3-D locations of the transmitters. %This is a reasonable assumption, as the differences in elevation amongst the devices of different users are relatively small compared to their distances as projected on the horizontal plane. 
Further, this assumption results in the 2-D blockage model that is elaborated next.

To model human body blockages, we associate each user's body with a circle of diameter $W$, as illustrated in Fig. \ref{fig:train_car}. These circles as well as the location of their centers are denoted by $B_i$. Like $X_i$, $B_i$ is represented as a complex number so that $B_i = |B_i| e^{j\angle B_i}$, where $|B_i|$ is the distance between the center of the $i^{th}$ human body blockage and the receiver, and $\angle B_i$ is the azimuth angle to $B_i$ from the reference receiver. In this blockage model, a transmitter $X_i$ is blocked if the direct path from $X_i$ to the reference receiver goes through the circle associated with any $B_j$ or if $X_i$ falls within the diameter-$W$ circle associated with any blockage $B_j$. The $i^{\mathrm{th}}$ user is associated with both a transmitter $X_i$ and a blockage $B_i$, and it is possible that transmitter $X_i$ is blocked by its corresponding blockage $B_i$. This is called \textit{self-blocking}, a phenomenon that was studied in \cite{Bai:Asilomar14}, in the context of 5G mmWave cellular system. 
If there are no blockages in the path from $X_i$ to the reference receiver, then we say that the path is \textit{line of sight} (LOS); otherwise, we say that it is \textit{non-LOS} (NLOS). We associate different channel parameters with LOS and NLOS paths, accounting for different path-loss and fading models inspired by measurements \cite{indoormeasurements, 15.3measurements}. In this paper, we assume that an interferer $i$ is potentially blocked from the reference receiver by $B_j, j\neq i$. Under this assumption that no signal is self-blocked, the following algorithm is used to determine which signals are blocked.
\begin{enumerate}
   \item Determine $\mathcal{L}$, the set of all transmitters $X_i$ that have no blockages $B_j, j \neq i$ within a distance of $W/2$; i.e.,
   \begin{eqnarray}
   \mathcal{L} = \left\lbrace X_i : |X_i - B_j| > \frac{W}{2}~ \forall j \neq i \right\rbrace, \label{Equation:Lset}
   \end{eqnarray} where $|X_i - B_j|$ is the distance along the horizontal plane between $X_i$ and $B_j$.
   \item Sort the blockages from closest to most distant, so that $|B_1| \leq |B_2| \leq ... \leq |B_K|$.   
   \item For each $i \in \{ 1,~2,~...,~ K \}$, compute the \textit{blocking cones (wedge in 2D)}\footnote{Strictly speaking, a blocking cone is an instance of a truncated cone because it does not extend to the origin.} $B_{C_i}$ as 
\begin{eqnarray}
{C_i} = \left\lbrace x \in \mathcal{A}: |x|>|B_i|,~ \angle B_i - \arcsin \left(\frac{W}{2|B_i|}\right)\leq \angle x \leq \angle B_i + \arcsin \left(\frac{W}{2|B_i|}\right)\right\rbrace. \label{Equation:BlockCone}
\end{eqnarray}
	\item For each $ \ell \in \mathcal{L}$, determine if $X_{\ell}$ is blocked by checking to see if it lies within any blocking cone; i.e., if 
	\begin{eqnarray}
	\phi_{\ell}\in \underset{\left\lbrace i: |B_i|< R_{\ell} \right\rbrace}{\bigcup} {C_i},
	\end{eqnarray} then $X_{\ell}$ is blocked. 
\end{enumerate} An illustration of the blocking cone discussed here is shown in Fig. \ref{fig:system_model}.

%\begin{figure}
%\centering
%\includegraphics[height = 1.5in, width = 3in]{blocking_cone_illus}
%\vspace{-.1in}
%\caption{Illustration of blocking cone $B_{C_i}$ and a blocked interferer $X_j \in B_{C_i}$.}
%\label{fig:blockingcone}
%\vspace{-.2in}
%\end{figure}

While the antenna gain pattern $G(\phi, \psi)$ is a complicated function of the azimuth angle $\phi \in \left[-\pi, \pi\right]$ and the elevation angle $\psi \in \left[-\frac{\pi}{2},\frac{\pi}{2}\right]$, to facilitate analysis, we use the three-dimensional sectorized antenna model as shown in Fig. \ref{fig:antenna_model}. We characterize the antenna array pattern with four parameters - the half-power beamwidth $\theta^{(\mathsf{a})}$ in the azimuth, the half-power beamwidth $\theta^{(\mathsf{e})}$ in the elevation, antenna gain $G$ within the half-power beamwidths (main-lobe) and gain $g$ outside it (side-lobe). We use the subscript $\mathsf{t}$ to denote an antenna parameter for a transmitter and subscript $\mathsf{r}$ for the receiver. For example, the main-lobe gain of the transmitter is $\Gt$ and that of the receiver is $\Gr$. Similarly, the side-lobe gains are denoted by $\gt$ and $\gr$. Because the system model and the analysis presented in this paper are general, substituting the appropriate values for the four parameters $G$, $g$, $\theta^{(\mathsf{e})}$ and $\theta^{(\mathsf{e})}$ into the expressions corresponding to the transmitters and the reference receiver enables the rapid evaluation of the SINR distribution. To compare performance in terms of directivity and gain based on practical antennas, we assume that a uniform planar square array (UPA) with half-wavelength antenna element spacing is used at the transmitters and the receiver. The number of antenna elements at the transmitter and receiver are denoted by $\Nt$ and $\Nr$, respectively. The antenna gain $G(\phi, \psi)$ of a UPA is modeled as a sectorized pattern as follows. The half-power beamwidths in the azimuth and the elevation are inversely proportional to $\sqrt{N}$ \cite{balanis:2012}. The main-lobe gain is taken to be $N$, which is the maximum power gain that can be obtained using $N$-element antenna array. Note that this is an approximation, though it is possible to design antennas to give near-flat response within the beamwidth with $G \propto N$.  The side-lobe gain is then evaluated so that the following antenna equation for constant total radiated power is satisfied \cite{balanis:2012}
\be
\underset{\hspace{-0.2in}-\pi}{\int^{\pi}} \underset{\hspace{-0.2in}-\frac{\pi}{2}}{\int^{\frac{\pi}{2}}} G(\phi, \psi) \cos(\psi) \mathsf{d}\psi \mathsf{d} \phi &=& 4\pi. \label{Equation:basic_antenna}
\ee
By using \eqref{Equation:basic_antenna}, we ensure the antenna arrays are passive components. The values for the half-power beamwidths (which are equal in both the azimuth and elevation for UPA), main-lobe and side-lobe gains for an $N$ element (i.e. $\sqrt{N}\times \sqrt{N}$) UPA are given in Table \ref{table:UPA}. When the number of antenna elements is one, we say that the UPA is omni-directional and, hence, the main-lobe and side-lobe gains are unity. This serves as a reference to compare the impact of antenna gain and directivity. As in \cite{bai:2014, SSingh_arX}, we assume that each interferer is transmitting with its main-lobe pointed in a random direction.

\begin{figure}
\centering
\begin{minipage}[t]{.55\textwidth}
\centering
\vspace{0pt}
\captionof{table}{Antenna parameters of a uniform planar square antenna}
\begin{tabular}{|c|c|} 
\hline 
Number of antenna elements & $N$  \\ 
\hline 
\hline
Half-power beamwidth, $\theta^{(\mathsf{a})} = \theta^{(\mathsf{e})}$ & $\frac{\sqrt{3}}{\sqrt{N}}$ \\ 
\hline
Main-lobe gain $G$ & $N$ \\ 
\hline 
Side-lobe gain $g$& $\frac{\sqrt{N}-\frac{\sqrt{3}}{2\pi} N \sin\left( \frac{\sqrt{3}}{2\sqrt{N}}\right)}{\sqrt{N}-\frac{\sqrt{3}}{2\pi}\sin\left(\frac{\sqrt{3}}{2\sqrt{N}}\right)}$ \\ 
\hline 
\end{tabular} \label{table:UPA}
\end{minipage} \hfill
\begin{minipage}[t]{.35\textwidth}
\centering
\vspace{0pt}
\includegraphics[height = 1.6in, width = 1.8in]{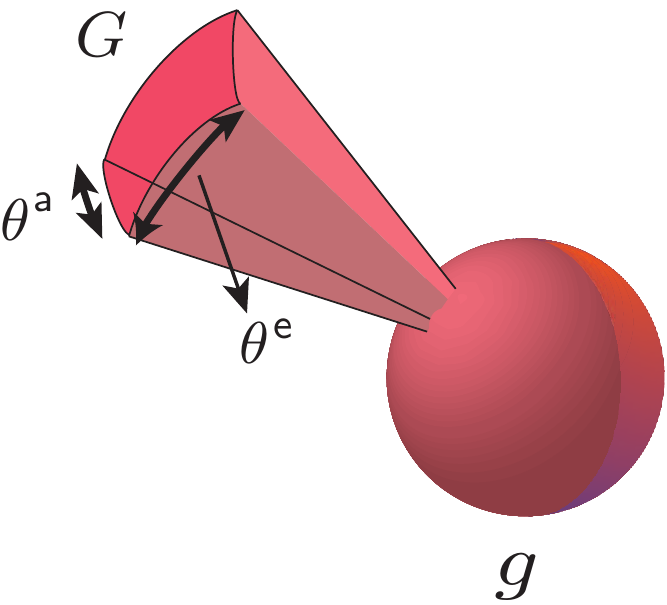}
\caption{Sectorized 3D antenna pattern.} \label{fig:antenna_model}%
\end{minipage}
\vspace{-.2in}
\end{figure}

We assume Nakagami fading for the wireless channels so that the power gain $h_i$ due to fading from $X_i$ to the reference receiver is Gamma distributed. We use $m_i$ to denote the Nakagami factor for the link from $X_i$ to the reference receiver, which assumes a value of $\mL$ for LOS and $\mN$ for NLOS \cite{bai:2014}. The path-loss exponent for $X_i$ is denoted as $\alpha_i$, where  $\alpha_i = \alphaL$ if $X_i$ is LOS and $=\alphaN$ if it is NLOS. There are different ways to define the signal-to-noise ratio (SNR) in a system with antenna arrays: with and without the antenna gains. We use $\sigma^{2}$ to denote the noise power divided by the reference transmitter power as measured at a reference distance excluding the antenna gains. While $\sigma^2$ is inversely proportional the SNR, we intentionally do not include the antenna gains into its computation, so that our results will naturally capture the SNR enhancement that accompanies the use of larger antenna arrays. The transmit power of $X_i$ is denoted as $P_i$.  Each interferer transmits with probability $\pt$, which is determined by the random-access protocol and user activity and is assumed to be the same for all interferers.

We assume that the reference communication link is always LOS. The reference link undergoes Nakagami fading with parameter $m_0 = \mL$ and has path-loss exponent $\alpha_0 = \alphaL$. Of course, it is possible that the reference user's body itself will create blockages on the reference link in a wearable network. When this occurs, it can be handled in our model by setting $m_0 = \mN$ and $\alpha_0 = \alphaN$. Capturing self-blockage of the reference link in a more refined model and incorporating the results into the analysis is an interesting topic for future work.

It is important to note that the boundaries of the finite area are assumed to be impenetrable, so there is no leakage of external interference into the finite area. Further, reflections due to the boundary and objects within the network are not explicitly incorporated in the model. They are accounted for only in a coarse way in the different LOS and NLOS model parameters, which ideally would be determined based on ray tracing or measurement results. The assumption of omitting reflections holds true in many scenarios where the boundaries of the finite area are made of poorly reflecting materials such as concrete or bricks.

%%%%%%%%%%%%%%%%%%%%%%%%%%%%%%%%%
\section{Interference Model}
\label{sec:Interfere_model}
Conditioned on the network (meaning the locations of the transmitters and blockages), we can find the complementary cumulative distribution function (CCDF) of the SINR (also called SINR coverage probability \cite{bai:2014}) by adapting the analysis in \cite{torrieri:2012, valenti:2013a}. The analysis that follows in Section \ref{ssec:coverage_prob} is very general since it can admit the individual interferers to have separate and independent values for the channel parameters - $\alpha_i$ and $m_i$, and does not require the LOS channel to have values $(\alphaL, \mL)$ and the NLOS channel to have values $(\alphaN, \mN)$. The assumption of fixing the channel parameters of the LOS and NLOS interferers yields tractable analytical expression for spatially averaged SINR coverage probability in Section \ref{sec:SpatialAvg}.

We define a discrete random variable $I_i$ for $i=\{1, ..., K\}$ that represents the relative power radiated by $X_i$ in the direction of the reference receiver.  With probability $(1-\pt)$, $X_i$ does not transmit at all, and hence $I_i = 0$.  Otherwise, the relative power will depend on whether or not the random orientation of $X_i$'s antenna is such that the reference receiver is within the main-lobe. We assume a uniform orientation of $X_i$'s antenna, so that the azimuth angle $\phi$ is uniform in $\left[0, 2\pi\right)$ and the elevation angle $\psi$ has a probability density function (pdf) $\frac{1}{2}\cos(\psi)$ in $\left[-\frac{\pi}{2},\frac{\pi}{2}\right]$. The pdfs can be derived by noting that the surface area element of a unit sphere is $\cos\left(\psi\right)\mathrm{d}\phi \mathrm{d}\psi$, a function of the elevation angle $\psi$. Thus the probability that the reference receiver is within the interferer's main-lobe is $\frac{\thetat^{(\mathsf{a})}}{2\pi} \sin\left(\frac{\thetat^{(\mathsf{e})}}{2}\right) = \pM$.  It follows that
\begin{eqnarray}
   I_i
   & = &
   \begin{cases}
        0 & \mbox{with probability $(1-\pt)$} \\
        \Gt & \mbox{with probability $\pt \pM$} \\
        \gt & \mbox{with probability $\pt \left( 1 - \pM \right)$}
   \end{cases}. \label{Equation:Ii}
\end{eqnarray}
Note that a similar approach was used in \cite{valenti:2013a} for modeling adjacent-channel interference in frequency hopping: when the interferer transmitted, one of two power compensations was applied depending on whether the interferer hopped into the same or an adjacent channel. In the wearable network context, we can justify randomizing the orientation angles of the interferers because: (1) the user itself may be randomly moving the orientation of its devices while using them, and (2) the user may have a wearable network with several devices with different orientations and random activity, though we assume the medium access protocol (MAC) of a user's wearable network allows only one of her devices to transmit at a time. It may be noted here that this kind of wearable network is still under development, so the exact MAC protocol has not yet been decided. We make a reasonable assumption that the network of a given user is coordinated such that only one device transmits at a time, while the devices of different users are not so coordinated and can therefore collide.  

Now, let us define the normalized power gain from $X_i$ to be 
\begin{eqnarray}
\label{Equation:Omega}   \Omega_i
   & = &
   \begin{cases}
      \frac{P_i}{P_0} \Gr R_i^{-\alpha_i} & \mbox{ if $-\frac{\thetar^{(\mathsf{a})}}{2} \leq \phi_i - \phi_0 \leq \frac{\thetar^{(\mathsf{a})}}{2}$} \\
      \frac{P_i}{P_0} \gr R_i^{-\alpha_i}& \mbox{ otherwise }
   \end{cases},     
\end{eqnarray}
where $\alpha_i = \alphaN$ if $X_i$ is NLOS and $\alpha_i = \alphaL$ if $X_i$ is LOS. This is the worst-case situation when $|\psi_0| \leq \frac{\thetar^{(\mathsf{e})}}{2}$. If $|\psi_0| > \frac{\thetar^{(\mathsf{e})}}{2}$, we have $\Omega_i = \frac{P_i}{P_0} \gr R_i^{-\alpha_i}, ~ \forall ~i$ which is a simpler trivial case. For the rest of the paper, we assume the non-trivial worst case and all the analysis presented hereafter extends easily for the trivial case. The SINR is
\begin{eqnarray}
    \gamma & = &
    \frac{ \Gt h_{0} \Omega_{0}}
    {\displaystyle \sigma^2 + \sum_{i=1}^{K} I_{i} h_{i} \Omega_{i} },  \label{Equation:SINR}
\end{eqnarray} where $\Omega_{0} = \Gr R_0^{-\alpha_0}$ is the normalized power gain from the reference transmitter, as we assume the reference transmitter is always within the main beam of the reference receiver. The effect of misalignment of beam in the reference link was considered at lower frequencies (e.g. UHF) in prior work \cite{Hanyu_beam, Wildman_beam}. In the wearable communication network context, however, since the distance of the reference link is short relative to the beamwidth of the antenna, pointing errors will not seriously degrade performance. (For instance, with our sectorized antenna model, the beam could be off by half the beamwidth without changing performance.)

%%%%%%%%%%%%%%%%%%%%%%%%%%%%%%%%%
\subsection{Coverage Probability}
\label{ssec:coverage_prob}
Denoting $\bf{\Omega}$ $= [\Omega_0, ..., \Omega_K]$, the coverage probability ${P}_{\mathsf{c}}(\beta, \bf{\Omega})$ for a given $\bf{\Omega}$ is defined as the CCDF of the SINR evaluated at a threshold $\beta$ and is given by
\begin{eqnarray}
P_{\mathsf{c}}(\beta, \bf{\Omega}) = \mathbb{P}\left[\gamma > \beta | \bf{\Omega}\right]. \label{Equation:Coverage}
\end{eqnarray}
Substituting \eqref{Equation:SINR} into \eqref{Equation:Coverage} and rearranging leads to
\be
P_{\mathsf{c}}(\beta, \bf{\Omega}) = \mathbb{P}\left[ \left. \mathsf{S} > \sigma^2 + \sum_{i=1}^{K} \mathsf{Y}_i \; \right| \; \bf{\Omega}\right], \label{Equation:SINR_exp} 
\ee
where $\mathsf{S} = \beta^{-1} \Gt h_0 \Omega_0$, and $\mathsf{Y}_i = I_i h_i \Omega_i.$
Conditioned on $\bf{\Omega}$, let $f_{\mathsf{\bf{Y}}}(\bf{y})$ denote the joint pdf of $(\mathsf{Y}_1, ..., \mathsf{Y}_K)$ and $f_{\mathsf{S}}(s)$ denote the pdf of $\mathsf{S}$. Then, \eqref{Equation:SINR_exp} can be written as
\begin{eqnarray}
P_{\mathsf{c}}(\beta, \bf{\Omega}) & = & \underset{\mathbb{R}^K}{\int ... \int} \left(\int_{\sigma^2+\sum_{i=1}^K y_i}^{\infty} f_{\mathsf{S}}(s)\mathsf{d}s \right) f_{\mathsf{\bf{Y}}}({\bf{y}})\mathsf{d}{\bf{y}}. \label{Equation:SINR_integral}
\end{eqnarray}
Defining $\beta_0 = {\beta m_0}/{\Gt \Omega_0}$ and assuming that $m_0$ is a positive integer,  the random variable $\mathsf{S}$ is  gamma distributed with pdf given by
\begin{eqnarray}
f_{\mathsf{S}}(s) &=& \frac{\left( \beta_0\right)^{ m_0}}{( m_0 - 1)!}s^{m_0 - 1} e^{- \beta_0 s} ,~ s \geq 0. \label{Equation:pdf_S}
\end{eqnarray}
Using \eqref{Equation:pdf_S}, the inner integral in \eqref{Equation:SINR_integral} is
\begin{eqnarray}
\label{Equation:inner_integral}
\overset{\infty}{\underset{\sigma^2+\sum_i^K y_i}{\int}} f_{\mathsf{S}}(s)\mathsf{d}s ~=~ e^{-\beta_0(\sigma^2 + \sum_{i=1}^K y_i )} \sum_{\ell = 0}^{m_0 - 1}\frac{(\beta_0\sigma^2)^\ell}{\ell!}\left(1 + \frac{1}{\sigma^2}\sum_{i=1}^K y_i\right)^\ell.
\end{eqnarray} Substituting \eqref{Equation:inner_integral} into \eqref{Equation:SINR_integral} leads to
\begin{eqnarray}
P_{\mathsf{c}}(\beta, {\bf{\Omega}}) = e^{-\beta_0\sigma^2}\sum_{\ell = 0}^{m_0 - 1}\frac{(\beta_0\sigma^2)^\ell}{\ell!} \underset{{\mathbb{R}}^K}{\int ... \int}e^{-\beta_0 \sum_{i=1}^K y_i}\left(1 + \frac{1}{\sigma^2}\sum_{i=1}^K y_i\right)^\ell f_{\mathsf{\bf{Y}}}({\bf{y}})\mathsf{d}{\bf{y}}. \label{Equation:shorthand_P}
\end{eqnarray}
Using the binomial theorem followed by multinomial expansion, 
\begin{eqnarray}
\left(1 + \frac{1}{\sigma^2}\sum_{i=1}^K y_i\right)^\ell =&  {\displaystyle \sum_{t = 0}^{\ell} \binom {\ell}{t}} \left(\frac{1}{\sigma^2} {\displaystyle \sum_{i=1}^K} y_i\right)^t =&  \sum_{t = 0}^{\ell} \binom {\ell}{t} \frac{t!}{\sigma^{2t}} \sum_{{\mathcal S}_t} \left( \prod_{i=1}^{K}\frac{y_i^{t_i}}{t_i!}\right), \label{Equation: multinomial}
\end{eqnarray} where the last summation is over the set ${\mathcal S}_t$ containing all length-$K$ non-negative integer sequences $\left\lbrace t_1, \ldots, t_K \right\rbrace$ that sum to $t$. This can be pre-computed and saved as a matrix as explained in \cite{torrieri:2012}. Substituting \eqref{Equation: multinomial} into \eqref{Equation:shorthand_P} gives 
\begin{eqnarray}
P_{\mathsf{c}}(\beta, \bf{\Omega}) &=& e^{-\beta_0\sigma^2}\sum_{\ell = 0}^{m_0 - 1}\frac{(\beta_0\sigma^2)^\ell}{\ell!}\sum_{t = 0}^{\ell}\binom {\ell}{t} \frac{t!}{\sigma^{2t}} \sum_{ {\mathcal S}_t }\underset{{\mathbb{R}}^K}{\int ... \int} \prod_{i=1}^{K}\frac{y_i^{t_i}}{t_i!}e^{-{\beta_0}y_i} f_{\mathsf{\bf{Y}}}({\bf{y}})\mathsf{d}{\bf{y}}. \label{Equation:exp_short_P}
\end{eqnarray}
Given $\bf{\Omega}$, the $\left\lbrace\mathsf{Y}_i\right\rbrace _{i = 1} ^{K}$ are independent. So, $f_{\mathsf{\bf{Y}}}({\bf{y}})$ may be written as $\prod_{i=1}^{K} f_{{\mathsf{Y}}_i}(y_i)$, where
\begin{eqnarray}
f_{{\mathsf{Y}}_i}(y_i) &=& \pt\left(\frac{m_i}{\Omega_i}\right)^{m_i}\frac{y_{i}^{m_i-1}}{\Gamma(m_i)}\left[\pM\frac{e^{-\frac{m_i y_i}{\Gt \Omega_i}}}{\Gt^{m_i}} + \left(1-\pM\right)\frac{e^{-\frac{m_i y_i}{\gt \Omega_i}}}{\gt^{m_i}}\right] u(y_i) + \left(1-\pt \right)\delta(y_i), \label{Equation: pdf_Yi}
\end{eqnarray} $\delta(y_i)$ is the Dirac delta function, and $u(y_i)$ is the unit step function. From the independence of the $\{\mathsf{Y}_i\}$,  \eqref{Equation:exp_short_P} may be written as 
\begin{eqnarray}
P_{\mathsf{c}}(\beta, \bf{\Omega}) &=& e^{-\beta_0\sigma^2}\sum_{\ell = 0}^{m_0 - 1}\frac{(\beta_0\sigma^2)^\ell}{\ell!}\sum_{t = 0}^{\ell}\binom {\ell}{t} \frac{t!}{\sigma^{2t}} \sum_{ {\mathcal S}_t }\left(\prod_{i=1}^{K} {\mathcal{G}}_{t_i}(\Omega_i)\right), \label{Equation:prod_form}
\end{eqnarray} where
\begin{eqnarray}
{\mathcal{G}}_{t_i}(\Omega_i)  &=&  \int_0^{\infty} \frac{y_i^{t_i}}{t_i!}e^{-{\beta_0}y_i} f_{{\mathsf{Y}}_i}(y_i)dy_i. \label{Equation:G_ti1}
\ee To evaluate \eqref{Equation:G_ti1}, we use the fact that $\frac{z^{k-1}e^{-\frac{z}{b}}}{b^k \Gamma(k)}$ is a probability density function (of a gamma-distributed random variable $Z$) with parameters $k,b > 0$, so that
\be
\int_0^{\infty} \frac{z^{k-1}e^{-\frac{z}{b}}}{b^k \Gamma(k)}dz = 1. \label{Equation:gamma_pdf}
\ee Accordingly, \eqref{Equation:G_ti1} simplifies to
\be
{\mathcal{G}}_{t_i}(\Omega_i) &=& \hspace{-0.1in} \pt \left(\frac{\Omega_i}{m_i}\right)^{t_i}\frac{\Gamma(m_i+ t_i)}{{t_i}!\Gamma(m_i)}\left[ \pM{\mathcal{Q}}_{t_i}(\Gt) + \left(1-\pM\right){\mathcal{Q}}_{t_i}(\gt) \right] + \left(1-\pt \right)\delta [t_i]. \label{Equation:G_ti}
\end{eqnarray} In \eqref{Equation:G_ti}, $\delta [t_i]$ is the function defined as
\begin{eqnarray}
\delta [t_i] &=& \begin{cases}
        1 & \mbox{if $t_i = 0$} \\
        0 & \mbox{if $t_i \neq 0$}
   \end{cases}\\
\text{and} ~~~
{\mathcal{Q}}_{t_i}(x) &=& x^{t_i}\left( 1 + \frac{\beta_0x \Omega_i}{m_i}\right)^{-(m_i + t_i)} \label{Equation:Qti}.
\end{eqnarray} The assumption of an integer value for $m_0$ is key to the derivation of the exact expression for the SINR coverage probability in \eqref{Equation:prod_form}. When $m_0$ is not an integer, such an exact evaluation is not possible to the best of our knowledge. Only an upper-bound using the results from \cite{Alzer:1997} can be obtained for a general real-valued $m_0$.

%%%%%%%%%%%%%%%%%%%%%%%%%%%%%%%%%
\subsection{Ergodic Spectral Efficiency}
\label{ssec:ergodic}
When the SINR is $\beta$, the spectral efficiency in bits per channel use is
\be
  \eta
  & = &
  \log_2 (1 + \beta).
\ee
The CCDF of the spectral efficiency is found by defining the equivalent events
\be
\{ \gamma > \beta ~|~ \bf{\Omega}\} 
& \Leftrightarrow & 
\underbrace{\{ \log_2(1+\gamma) > \eta  ~|~ \bf{\Omega} \}}_{ \{ \gamma > 2^{\eta} - 1  ~|~ \bf{\Omega}\} }.
\ee
The event on the left corresponds to the coverage probability $P_\mathsf{c}(\beta, \bf{\Omega})$, while the event on the right corresponds to the CCDF of the spectral efficiency, $P_\mathsf{\eta}(\eta, \bf{\Omega})$, also called the \emph{rate coverage probability} for a given $\bf{\Omega}$.  Since equivalent, the two events have the same probability, and it follows that
\be
  P_\mathsf{\eta}(\eta, \bf{\Omega})
  & = &
  P_\mathsf{c} \left( 2^{\eta} - 1, \bf{\Omega} \right).
\ee
See also Lemma 5 of \cite{bai:2014}.

Using the fact that, for a non-negative $X$, $\mathbb{E}[X] = \int_0^\infty (1-F(x)) \mathsf{d} x$ (see (5-33) in \cite{papoulis:2002}), the {\em ergodic spectral efficiency} conditioned on $\bf{\Omega}$ can be found from
\be
  \mathbb{E}[\eta]
   = &
  \int_0^\infty P_\mathsf{c} \left( 2^{\eta} - 1 , \bf{\Omega} \right) \mathsf{d} \eta  = &
  \frac{1}{\log(2)}\int_0^\infty \frac{ P_\mathsf{c} \left( \beta, \bf{\Omega} \right)}{1+\beta} \mathsf{d} \beta,
\ee
where the last step uses the change of variables $\beta = 2^{\eta} - 1 \rightarrow \mathsf{d}\eta = \frac{1}{\log(2)}\mathsf{d}\beta/(1+\beta)$.

In practice, there is a maximum and minimum rate, and hence, a maximum and minimum SINR thresholds $\beta_\mathsf{max}$ and $\beta_\mathsf{min}$, respectively.  This maximum may be imposed by the modulation order of the constellation used and distortion limits in the RF front end while minimum due to the receiver sensitivity.  In this case, the limits of the integral are $\beta_\mathsf{min}$ and $\beta_\mathsf{max}$, and
\be \label{Equation:ergSpecEffeciency}
  \mathbb{E}[\eta]
  & = &
  \int_{\beta_\mathsf{min}}^{\beta_\mathsf{max}} \frac{ P_\mathsf{c} \left( \beta, \bf{\Omega} \right)}{\log(2)(1+\beta)} \mathsf{d} \beta.
\ee The quickest way to compute \eqref{Equation:ergSpecEffeciency} is to simply compute $P_\mathsf{c} \left( \beta, \bf{\Omega} \right)$ for a finely spaced $\beta$ and then use the trapezoidal rule to numerically solve the integral.

%%%%%%%%%%%%%%%%%%%%%%%%%%%%%%%%%
\section{Numerical Results for fixed geometry}
\label{sec:Num_results}
In this section, we provide numerical results for coverage probability and ergodic spectral efficiency. The users are located at fixed locations, but to enable a comparison against random topologies (see Section \ref{sec:sim_results}), their placement is confined to an annulus $\mathcal{A}$ having inner radius $\rin$ and outer radius $\rout$. Conditioned on the fixed locations of the interferers and the blockages, the exact expression for the SINR coverage probability can be deriverd using \eqref{Equation:prod_form}. We assume there are $K$ interfering transmitters, neglect self-blocking, and assume that the blockage and transmitter associated with each user are co-located; i.e., $B_i = X_i$ for each $i$. It is assumed that the $P_i$ are all the same; i.e., all transmitters transmit at the same power.

\begin{table}
\centering
\caption{Antenna Parameters} \label{table:antenna_param}
\vspace{-0.1in}
\begin{tabular}{|c|c|c|c|}
\hline 
Number of antenna elements & 1 & 4 & 16 \\ 
\hline 
\hline 
Half-power beam width in the elevation and azimuth (in degrees) & 360 & 49.6 & 24.8 \\ 
\hline 
Main-lobe gain (in dB)& 0 & 6 & 12 \\ 
\hline 
Side-lobe gain (in dB)& 0 & -0.8839 & -1.1092 \\ 
\hline 
\end{tabular}
\vspace{-0.2in} 
\end{table}

\begin{table}
\centering
\caption{Parameters used to obtain numerical results for fixed geometry } \label{table:num_param}
\vspace{-0.1in}
\begin{tabular}{|c|c|c|}
\hline 
Parameter & Value & Description\\ 
\hline
\hline
$R_0$ & 0.3 m & Reference link length\\
\hline
$\phi_0, \psi_0$ & $0^o$ & Antenna main-lobe orientation of the reference receiver\\
\hline
$\mL$ & 4 & Nakagami parameter for LOS link\\ 
\hline 
$\mN$ & 2 & Nakagami parameter for NLOS link\\ 
\hline 
$\alphaL$ & 2 & Path-loss exponent for LOS link\\ 
\hline 
$\alphaN$ & 4 & Path-loss exponent for NLOS link\\ 
\hline 
$W$ & 0.3 m & Width of the human-body blockages \\ 
\hline 
$\sigma^2$ & -20 dB & Noise power normalized by reference transmitter power\\ 
\hline 
$K$ & 36 & Number of potential interferers\\ 
\hline 
\end{tabular}
\vspace{-0.2in}
\end{table}

The values of the antenna half-power beamwidths, main-lobe and side-lobe gains are summarized in Table \ref{table:antenna_param}. Note that it is possible to get desired side-lobe isolation by carefully designing the array response via windowing similar to filter design \cite{balanis:2012}. This would also add complexity to the array design and configuration. Since power and heating issues are critical for wearable devices, it is yet to be determined if such techniques would indeed be considered in future gadgets. To quantify the effect of antenna directivity and as an example, we chose a uniform planar array described in Table \ref{table:UPA}. The network and signal parameters used to obtain the results in this section are summarized in Table \ref{table:num_param}. The Nakagami parameters and the path-loss exponents assumed are the ones used in \cite{bai:2014}. For simplicity and to ensure the interferers are uniformly spread out in the network region, we let the user locations to be on a $n \times n$ square lattice restricted to the annulus $\mathcal{A}$. The network region under this assumption is shown in Fig. \ref{fig:fixed_grid} where the $7 \times 7$ grid locations and the user locations in the network region $\mathcal{A}$ are shown.  We let the minimum distance of two nodes in the grid be $2R_0$. For example, when the lattice points are separated by $0.6$ m ($R_0 = 0.3$ m as in Table \ref{table:num_param}) and $n$ = 7, we get $K$ = 36 with $\rin = 0.3$ m and $\rout = 2.1$ m, which corresponds to an interferer density $\lambda = 2.25$ (passengers/m$^2$), a typical density scenario%\footnote{A reasonable value for a unit is a foot, or about 0.3 m} 
that approximates the peak-hour passenger load in urban train cars \cite{trainDensity}. 

\begin{figure}
\centering
\subfigure[The locations of the users in a uniform grid of size 7 $\times$ 7 restricted to an annulus. The twelve users located outside the circle are deleted from the network.]{
\includegraphics[totalheight=3.2in, width=4.8in]{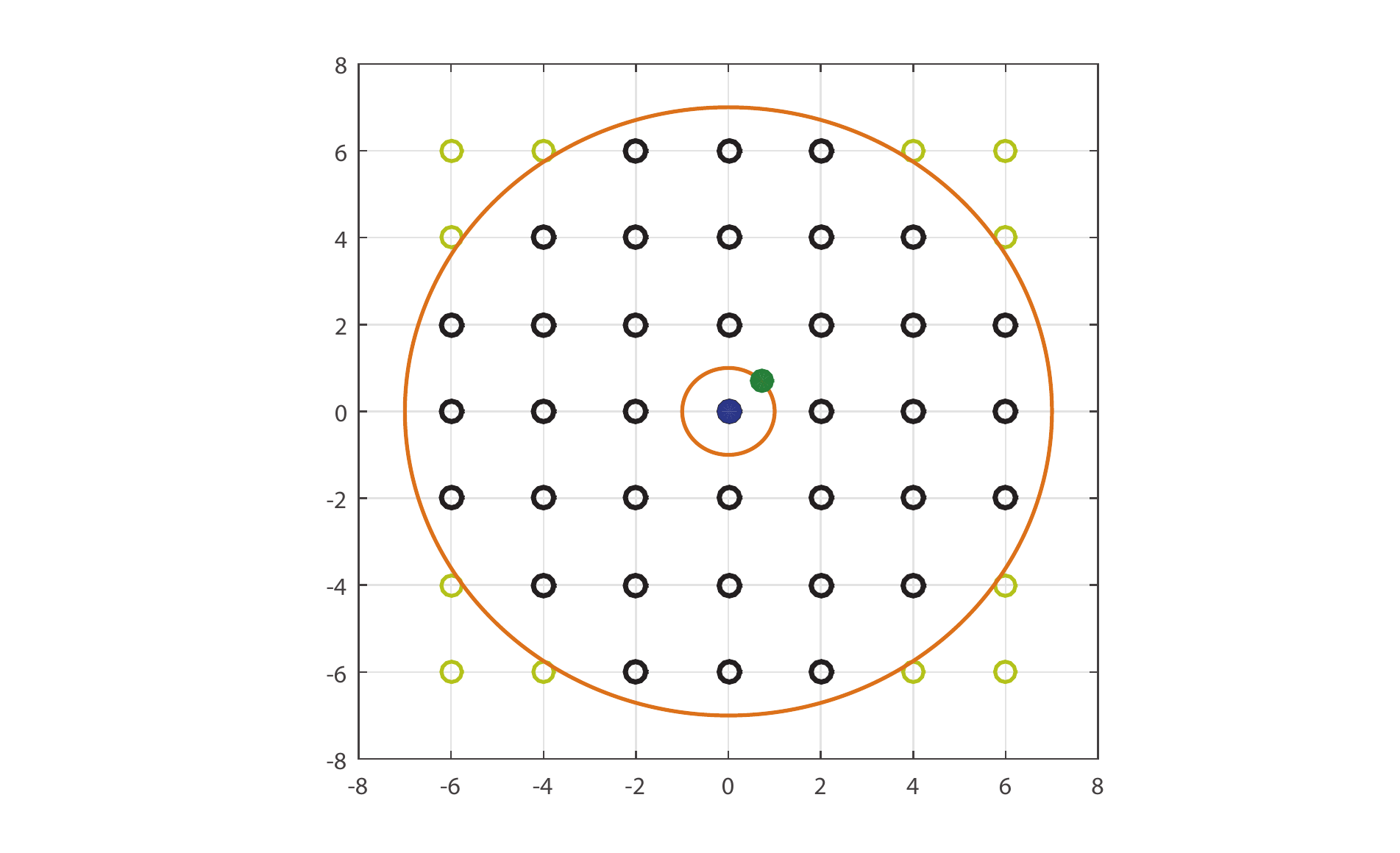}    
\label{fig:fixed_grid}
\vspace{-0.2in}        
}
\subfigure[The blocking cones associated with the blockages and the blocked users (filled circles).]{
\includegraphics[totalheight=3.2in, width=4.8in]{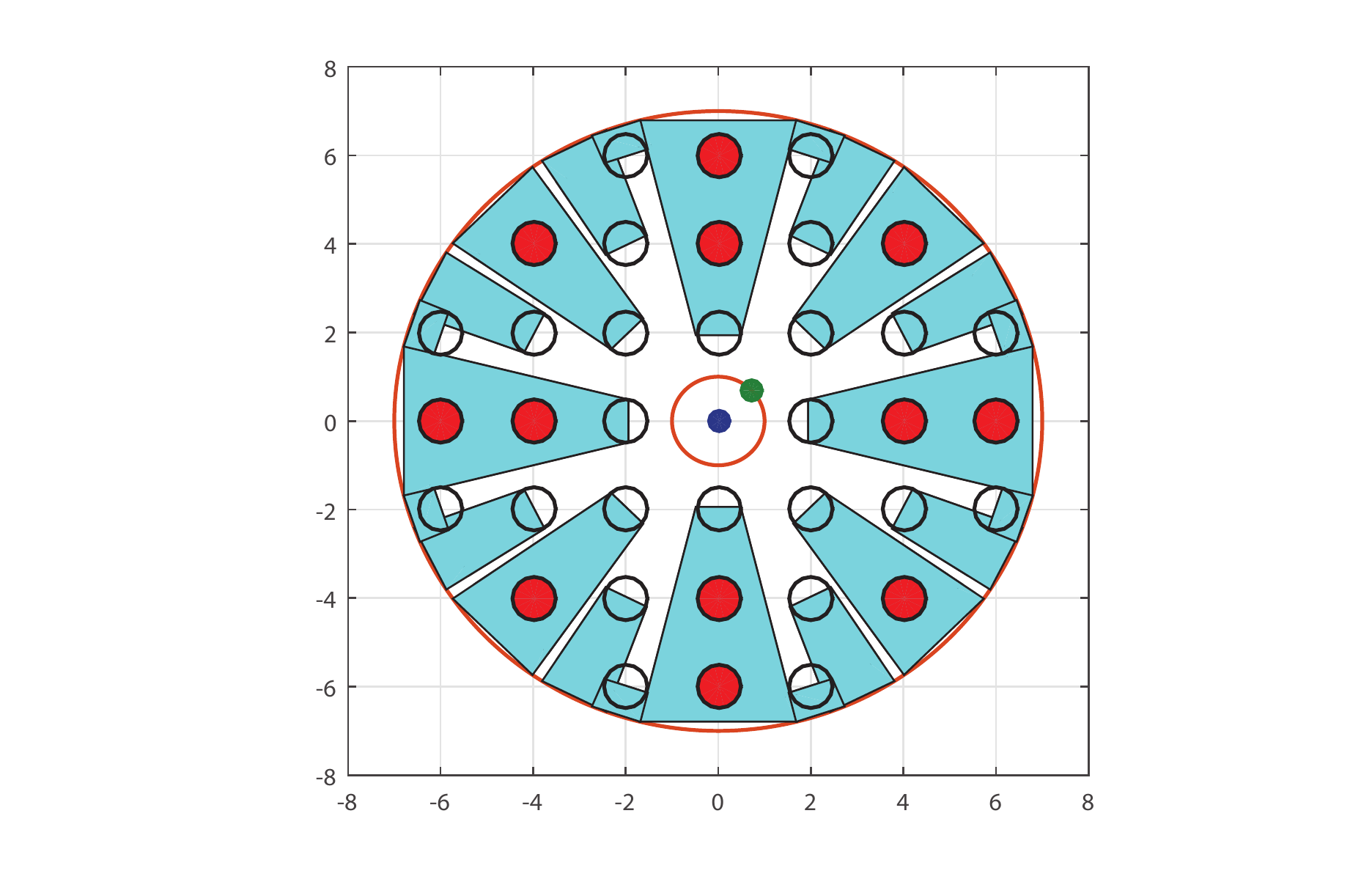}     
\label{fig:fixed_gridCones}
\vspace{-0.2in}    
}
\caption{The fixed geometry considered in Section \ref{sec:Num_results} and the blocking cones associated with the users. The reference receiver and the projection of the transmitter onto $\mathcal{A}$ are shown in blue and green, respectively.}    
\label{fig:fixed_geo}
\end{figure}

Fig. \ref{fig:fixed_geo} shows users placed according to Fig. \ref{fig:fixed_grid} along with the blocking cones (Fig. \ref{fig:fixed_gridCones}) assuming that each user is associated with a blockage of width $W=0.3$ m. The blocked users are indicated by filled circles. We next provide numerical results for this fixed geometry. The dependence of coverage probability on the transmission probability $\pt$ of the interferers for a fixed transmitter and receiver antenna array configuration is shown in Fig. \ref{fig:coverageVpt} for the case when the transmitters and the receiver use omni-directional antenna. It is seen that, as expected, a higher value of $\pt$ leads to lower coverage probability for a given SINR threshold. We observe similar results for other antenna configurations as well. 
\begin{figure}
\centering
\includegraphics[totalheight=3.2in, width=3.5in]{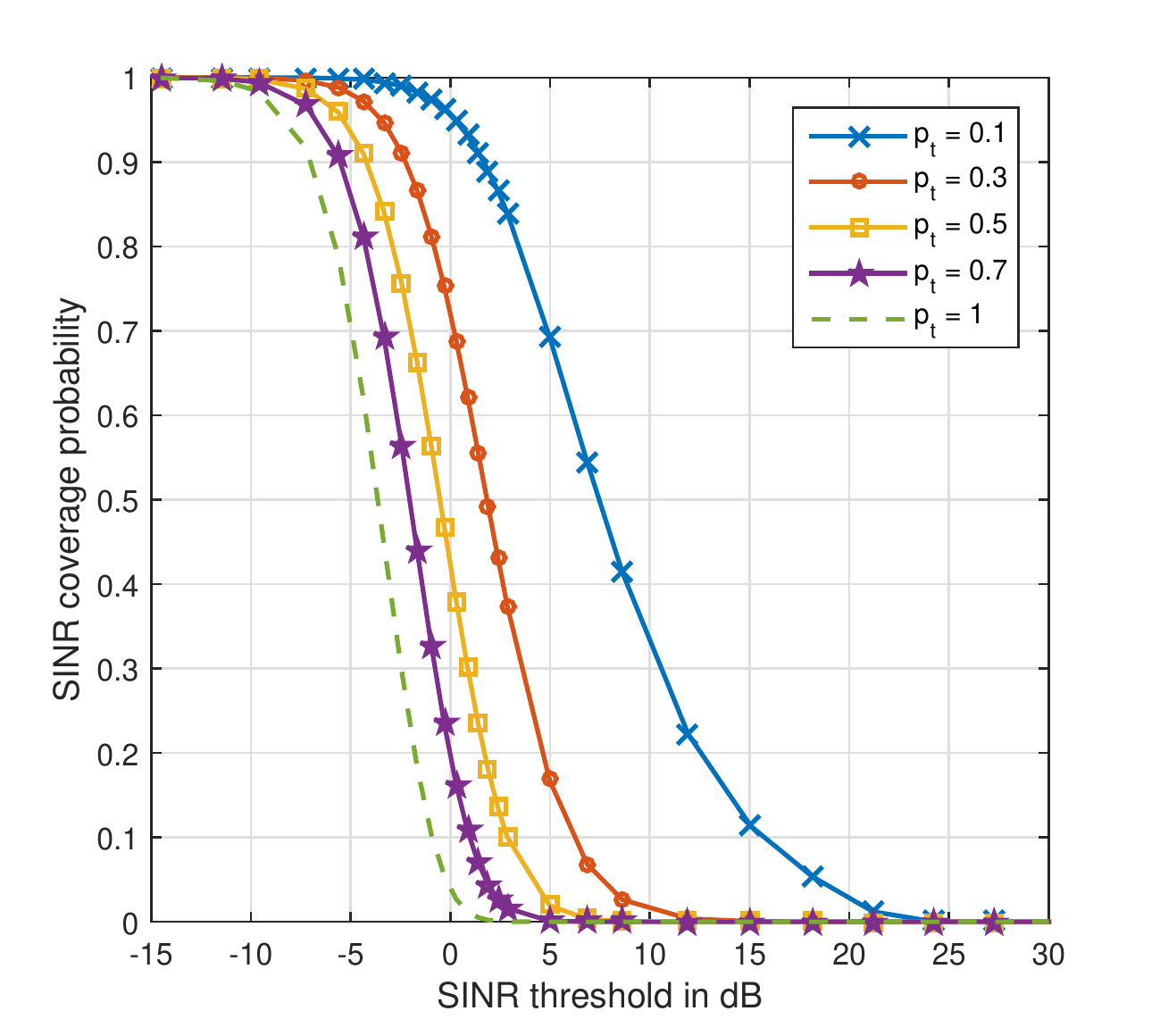}
\vspace{-0.2in}  
\caption{SINR coverage probability when the users are placed in the fixed positions indicated in Fig. \ref{fig:fixed_grid} for different transmission probabilities $\pt$ $\Nt = \Nr = 1$. Larger $\pt$ results in smaller coverage probability.}     
\label{fig:coverageVpt}        
\end{figure}

The CCDF of spectral efficiency for different antenna configurations is shown in Fig. \ref{fig:rate_NtXNr} for a given random-access probability. Here we let $\pt = 1$. Clearly, using more antennas at the transmitters and the receiver results in significant improvement in the rate. This is because larger antenna arrays provide more directed transmission and reception, thus improving the SINR due to the increased antenna gains of the reference link as well as the reduced beamwidth of the interfering receivers, which reduces the likelihood that the reference receiver falls within a randomly oriented reciever's main-lobe. 
The ergodic spectral efficiency for various antenna configurations when $\pt = 1$ is summarized in Table \ref{table:avgrate}. It can be seen that having larger $\Nt$ is more advantageous than having a larger $\Nr$ for the fixed geometry considered in this section. This can also be seen from Fig. \ref{fig:rate_contour}, which is a contour plot of the ergodic spectral efficiency as a function of $\Nt \geq 4$ and $\Nr \geq 4$ for a random-access probability $\pt = 0.5$ and found by interpolating the computed values of the ergodic spectral efficiency for all integer values of $\Nt$ and $\Nr$ from 4 to 16. We attribute the asymmetric behavior with respect to $\Nt$ and $\Nr$ in Fig. \ref{fig:rate_contour} to the fact that the interferers have smaller probability of pointing their main-lobes to the reference receiver when $\Nt$ is large.

\begin{figure}[ht]
\centering
\includegraphics[totalheight=3.2in, width=3.5in]{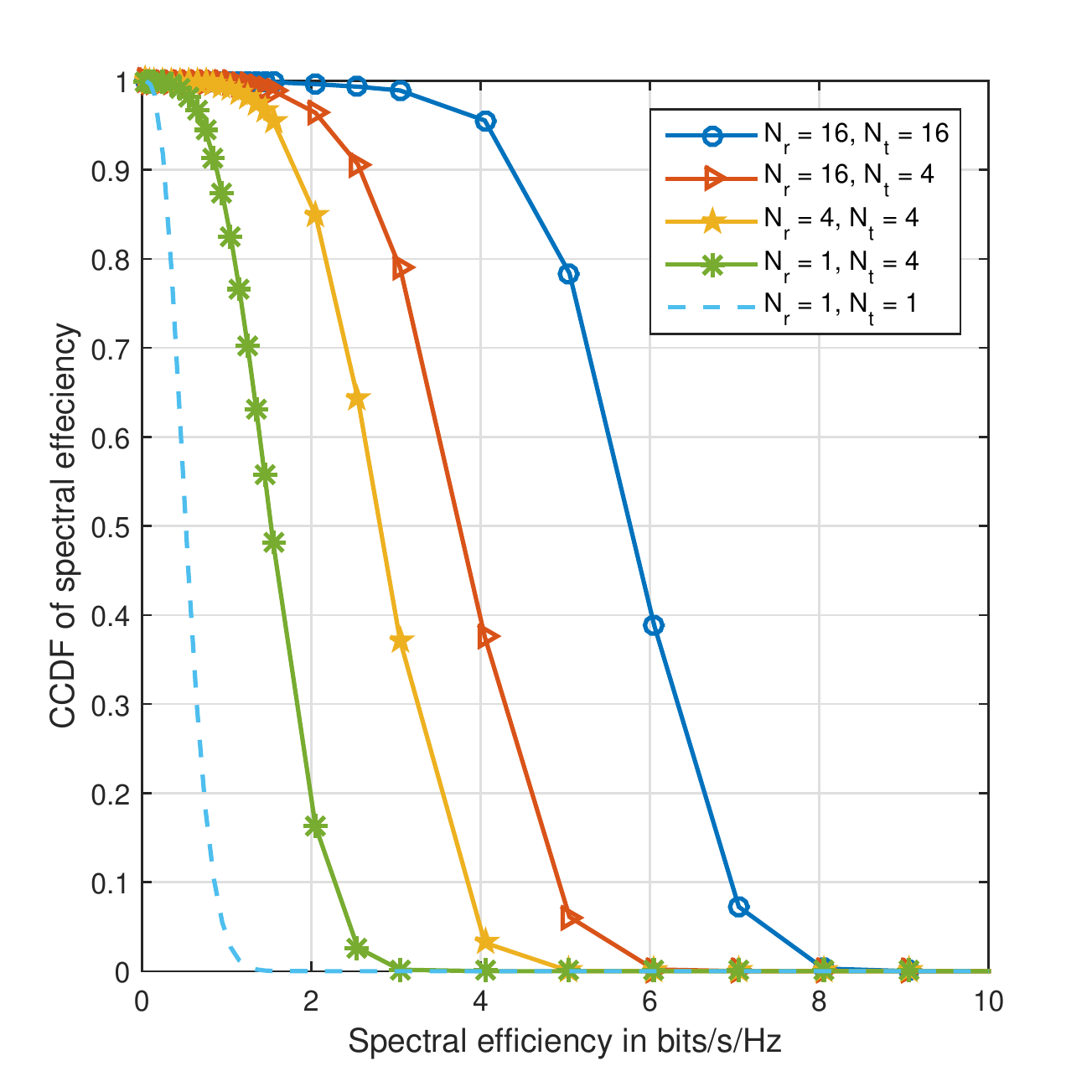}
\vspace{-0.2in}  
\caption{CCDF of spectral efficiency for different antenna configurations when $\pt = 1$ and the fixed network geometry in Fig. \ref{fig:fixed_grid}. Spectral efficiency is improved significantly with more antennas.}
\label{fig:rate_NtXNr}        
\end{figure}

\begin{table}
\centering
\caption{Ergodic spectral efficiency for various antenna configurations} \label{table:avgrate}
\vspace{-0.1in}
\begin{tabular}{|l||*{3}{c|}}\hline
\backslashbox{$\Nt $}{$\Nr $} & \makebox[2em]{1} & \makebox[2em]{4} & \makebox[2em]{16} \\ \hline
\hline
1 & 0.1762 & 0.8710 & 1.5481 \\
\hline 
4 & 1.0880 & 2.3282 & 3.2820 \\
\hline 
16 & 2.6734 & 4.2190 & 5.2850 \\
\hline 
\end{tabular}
\vspace{-0.2in}
\end{table} %\beta_min = -5dB, \beta_max  = 40 dB
%\begin{table}
%\caption{Ergodic spectral efficiency for various antenna configurations} \label{table:avgrate}
%\begin{center}
%\begin{tabular}{|c|c|c|c|}
%\hline 
%{$\Nt$} & {$\Nr$} & Ergodic spectral efficiency (bits/s/Hz)\\ 
%\hline 
%1 & 1 & 0.4943 \\ 
%\hline 
%1 & 4 & 0.7920 \\ 
%\hline
%1 & 16 & 1.7527\\ 
%\hline
%4 & 1 & 0.7971\\ 
%\hline
%4 & 4 &  1.1997\\ 
%\hline
%4 & 16 & 2.3761\\
%\hline
%16 & 1 & 2.1459\\
%\hline
%16 & 4 & 2.8030\\
%\hline
%16 & 16 & 4.4849\\
%\hline
%\end{tabular} 
%\end{center}
%\end{table}

\begin{figure}
\centering
\includegraphics[totalheight=3.2in, width=3.5in]{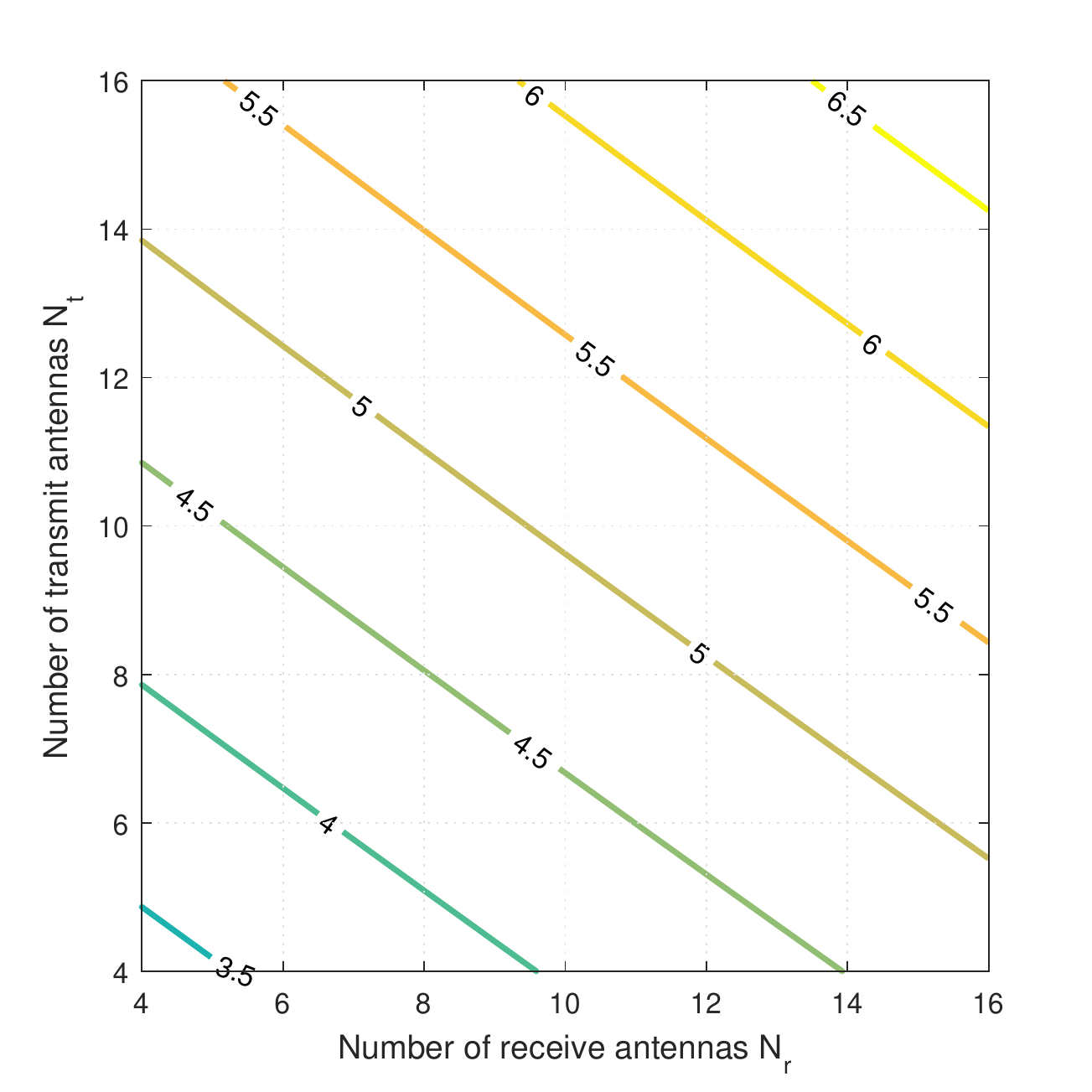}
\vspace{-0.2in} 
\caption{Contour plot of ergodic spectral efficiency as a function of $\Nt$ and $\Nr$ when $\pt = 0.5$ for the fixed interferer placement according to Fig. \ref{fig:fixed_grid}.}     
\label{fig:rate_contour}        
\end{figure}

%The analysis in Section \ref{ssec:coverage_prob} can be used to characterize performance at different locations in the network. Hence, this approach can also be used to determine network performance. As an example, we show the SINR coverage probability for an SINR threshold of 5 dB in Fig. \ref{fig:coverage_heat} for the network considered in Fig. \ref{fig:fixed_grid}. We assume $\Nt = 1$, $\Nr = 16$ and let $\pt = 1$. It can be seen that due to the finite geometrical restriction of the network, the performance varies considerably based on the reference user's location. 
%
%\begin{figure}
%\centering
%\includegraphics[totalheight=3.2in, width=3.2in]{network_coverage_new}
%\vspace{-0.5in}  
%\caption{Heat map showing the variation in coverage probability for an SINR threshold of 5 dB based on location of users as per Fig. \ref{fig:fixed_grid}. We assume  $\Nt =1$, $\Nr = 16$ and $\pt = 1$. The SINR coverage probability varies considerably based on the location.}     
%\label{fig:coverage_heat}        
%\end{figure}

\section{Spatial Averaging for Random Geometries}
\label{sec:SpatialAvg}
Infinite-sized networks are usually analyzed by assuming the interferers are drawn from a point process, then determining the coverage and rate of a typical user by averaging over the network geometry. Intuitively, this can be thought of as the performance seen by a user that wanders throughout the network, and thus sees many different network topologies. In this section, we outline the approach we take to analyze the interference in a finite sized mmWave-based device-to-device network when the users are located at random locations. While in reality, users are generally spaced far enough apart that their bodies don't overlap, for mathematical tractability we assume that the users are independently placed (which include cases with overlaps).

The spatially averaged CCDF of the SINR can be derived by taking the expectation of $P_{\mathsf{c}}(\beta, \bf{\Omega})$ (\eqref{Equation:prod_form} in Section \ref{sec:Interfere_model}) with respect to $\mathbf{\Omega}$
\be
P_{\mathsf{c}}(\beta) = & \mathbb{P}\left[\gamma > \beta \right] =& \mathbb{E}_{\bf{\Omega}}[P_{\mathsf{c}}(\beta, \bf{\Omega})] \label{Equation:CCDF_ways}.
\ee This can be performed in two ways: (1) through the use of simulation, and (2) analytically. In the former method, we can determine the coverage and rate as follows.  We randomly place the $K$ potentially interfering users drawn from a binomial point process (BPP) and compute the corresponding coverage and rate for each network realization. This is repeated a large number of times to obtain the spatial average.  While in the limit of an infinite number of trials, this approach provides the exact spatially averaged performance, the downside is that it is computationally expensive to obtain. The second method is similar to the approach in \cite{valenti:2014}, which however only considered Rayleigh fading for the links. The spatially averaged outage probability is found in this approach by unconditioning the results we obtained in Section \ref{sec:Interfere_model} that were conditioned on the location of the interferers and the blockages. Using this approach, we develop closed-form expression for the spatially averaged CCDF of the SINR, which is then validated against the results obtained via simulation.

\subsection{Assumptions}
\label{sec:Assumptions}
Taking the expectation is complicated by a number of factors that arise primarily due to the coupling of interferers and blockages, since each user is not only a potential source of interference, but is also a potential source of blockage. To make the problem more tractable, we adopt a sequence of assumptions, with each assumption building upon the previous one. Simulation results show the validity of the assumptions.

\emph{Assumption 1:}  {\bf The locations of the blockages and interferers are related by an orbital model.} Even if a user $B_i$ (which also denotes blockage) is in a fixed location, its transmitter $X_i$ could be randomly positioned around it. Hence, we specify the location of $X_i$ in the 2-D plane relative to $B_i$ by placing it randomly on a circle of radius $d > W/2$ and center $B_i$. Self-blocking is now inherent in the model, i.e., if $X_i$ is behind $B_i$, then it is blocked. We refer to this model as the \textit{orbital model}. This is illustrated in Fig. \ref{fig:orbital_model}. 

\begin{figure}
\centering
\includegraphics[totalheight=1.2in,width=2.4in]{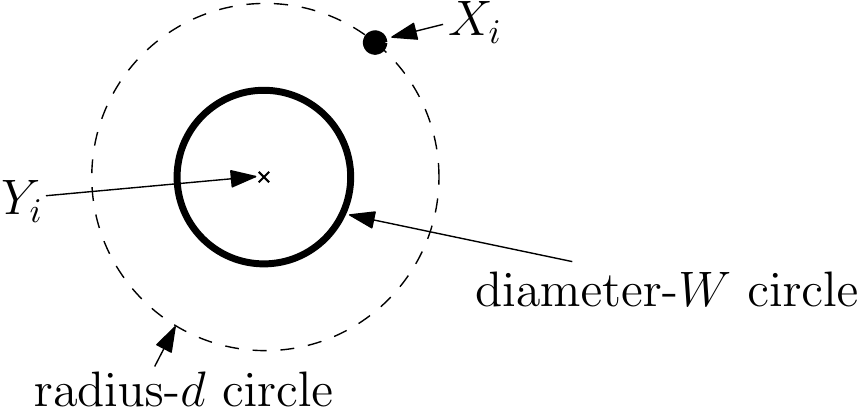}
\vspace{-0.2in}  
\caption{An illustration of the orbital model with the locations of blockage $B_i$ and the interferer $X_i$.}     
\label{fig:orbital_model}
\vspace{-0.2in}        
\end{figure}

\emph{Assumption 2:}  {\bf The locations of the blockages and interferers are drawn from independent point processes.} Though this assumption is not by itself that useful, it is a stepping stone towards a tractable analysis. With this assumption we can still obtain the corresponding coverage and rate using the aforementioned simulation, only now the simulation can lay out $K$ interferers and $K$ blockers independently.  

\emph{Assumption 3:} {\bf The blockage states of the interferers are independent.}  This assumes that there is no correlation in the blockage process, even though in reality a transmitter that is close to a blocked transmitter is likely to also be blocked.  With this assumption we first determine the blockage probability $\pb(r)$, which gives the probability of blockage arising from other users as a function of distance $r$ to the reference receiver.  We can determine $\pb(r)$ either empirically (through running simulations of the blockage process) or by using results from random shape theory \cite{Bai2014}.  Then, having established $\pb(r)$, we run a new simulation whereby we first place the interferers according to a BPP, then we independently mark each interferer as being blocked with probability $\pb(r)$. 
Note that now, we need not place the blockages in the simulation.

\emph{Assumption 4:} {\bf All interferers beyond some distance $\RB$ are NLOS and those closer than $\RB$ are LOS.}  Here, we replace the irregular and random LOS boundary with an equivalent ball.  The value of $\RB$ can be found by matching the first moments (Criterion 1 of \cite{bai:2014}), or alternatively, it could be found by matching the average rate.

\subsection{Analysis of Blocking Probability}
\begin{lemma}
\label{lemma:snowboard}
When the network region $\mathcal{A}$ is an annulus with inner radius ${\rin}$ and outer radius ${\rout}$ and blockages have diameter $W$, the probability that an interferer at distance $r$ from the reference receiver is blocked by any of the $K$ blockages that are independently and uniformly distributed over $\mathcal{A}$ is 
\begin{eqnarray}
p_{\mathsf{b}}(r)
   & = &
   \begin{cases}
        1 - \left( 1 - \frac{rW + \frac{\pi W^2}{8} - \mu }{|\mathcal{A}|}  \right)^{K} & \mbox{if ${\rin} \leq r \leq {\rout} - \frac{W}{2}$} \\
         1 - \left( 1 - \frac{rW - \mu + \nu }{|\mathcal{A}|}  \right)^{K} & \mbox{if ${\rout} - \frac{W}{2} \leq r \leq {\rout}$} 
   \end{cases}, \label{Equation:pbr}
\end{eqnarray} 
where 
\begin{eqnarray}
\mu &= &\frac{W}{2}\sqrt{{\rin}^2 - \left(\frac{W}{2}\right)^2}+{\rin}^2\arcsin\left(\frac{W}{2{\rin}}\right)\\
\nonumber
\nu &=& \left(\frac{W}{2}\right)^2\arcsin\left(\frac{{\rout}^2 - \left(\frac{W}{2}\right)^2 - r^2}{rW}\right)+ {\rout}^2\arccos\left(\frac{{\rout}^2 - \left(\frac{W}{2}\right)^2 + r^2}{2r{\rout}}\right)\\
&&~~~~~-2\sqrt{s(s-r)(s-\frac{W}{2})(s- \frac{\rout}{2})}; s = \frac{{\rout}+r+W/2}{2}.
\end{eqnarray}
\end{lemma}

\begin{proof}
See Appendix \ref{Appdx:proof_snowboard}
\end{proof}
The distance-dependent blockage probability $p_{\mathsf{b}}(r)$ is shown for an annulus of inner radius $\rin =1$ and outer radius $\rout = 7$ in Fig. \ref{fig:p_br} along with that obtained via simulation. The simulation results closely match the analytical approximations derived using Assumptions 1-3. Also, as expected, the probability that an interferer is blocked increases as the distance $r$ between the interferer and the reference receiver is increased. The dependence of $p_{\mathsf{b}}(r)$ on $K$, the number of interferers is shown when $W=1$ in Fig. \ref{fig:pbVsK} and the dependence on the width of the blockages $W$ is shown when $K=36$ in Fig. \ref{fig:pbVsW}. It is seen that with larger values for $K$ and $W$ the blockage probability is more for any given distance $r$ from the reference receiver.
\begin{figure}
\centering
\subfigure[$p_{\mathsf{b}}(r)$ vs $r$ for different values of $K$ when $W = 1$.]{
\includegraphics[totalheight=3.2in, width=3.5in]{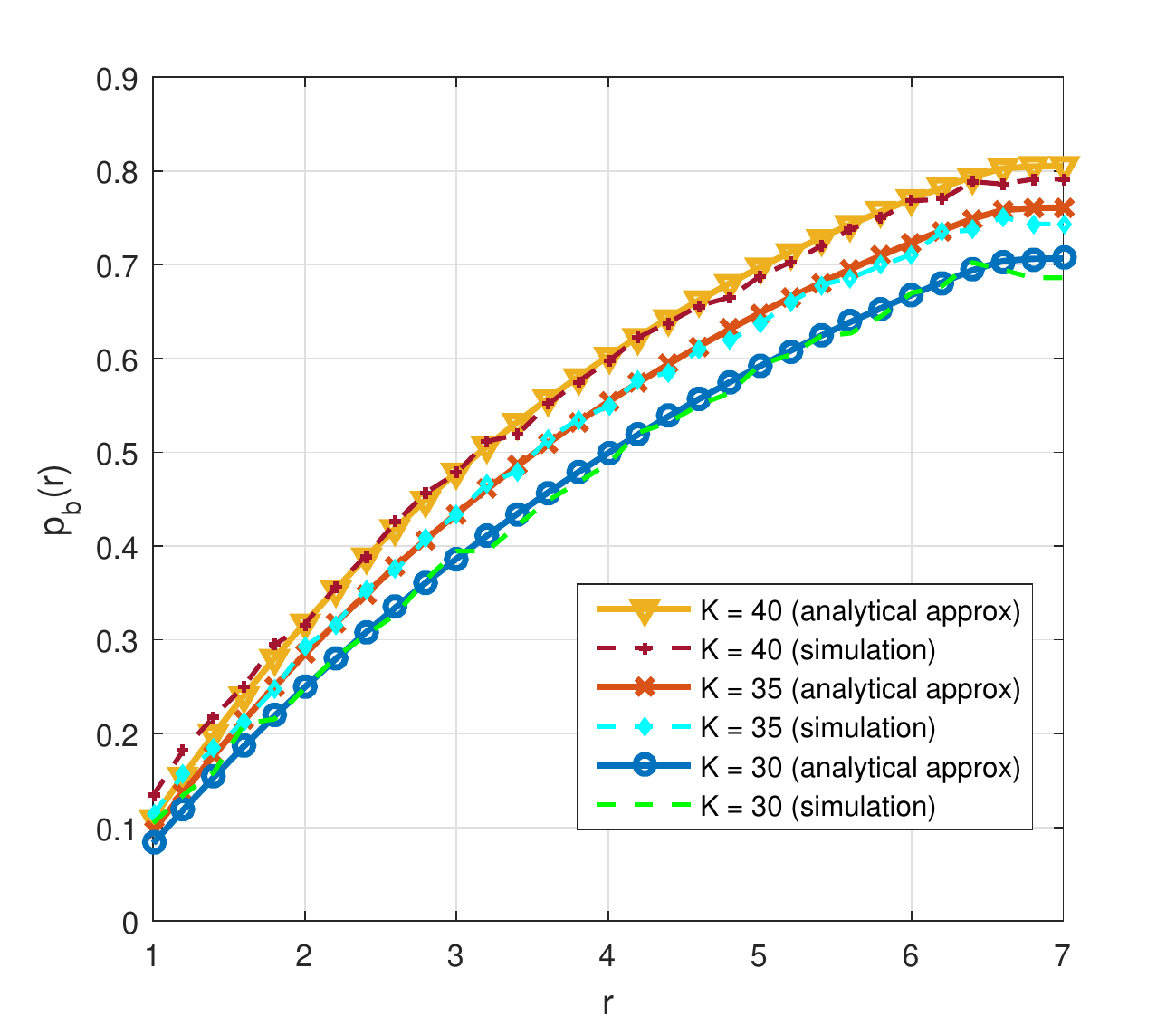}
\label{fig:pbVsK}
}
\subfigure[$p_{\mathsf{b}}(r)$ vs $r$ for different values of $W$ when $K=36$.]{
\includegraphics[totalheight=3.2in, width=3.5in]{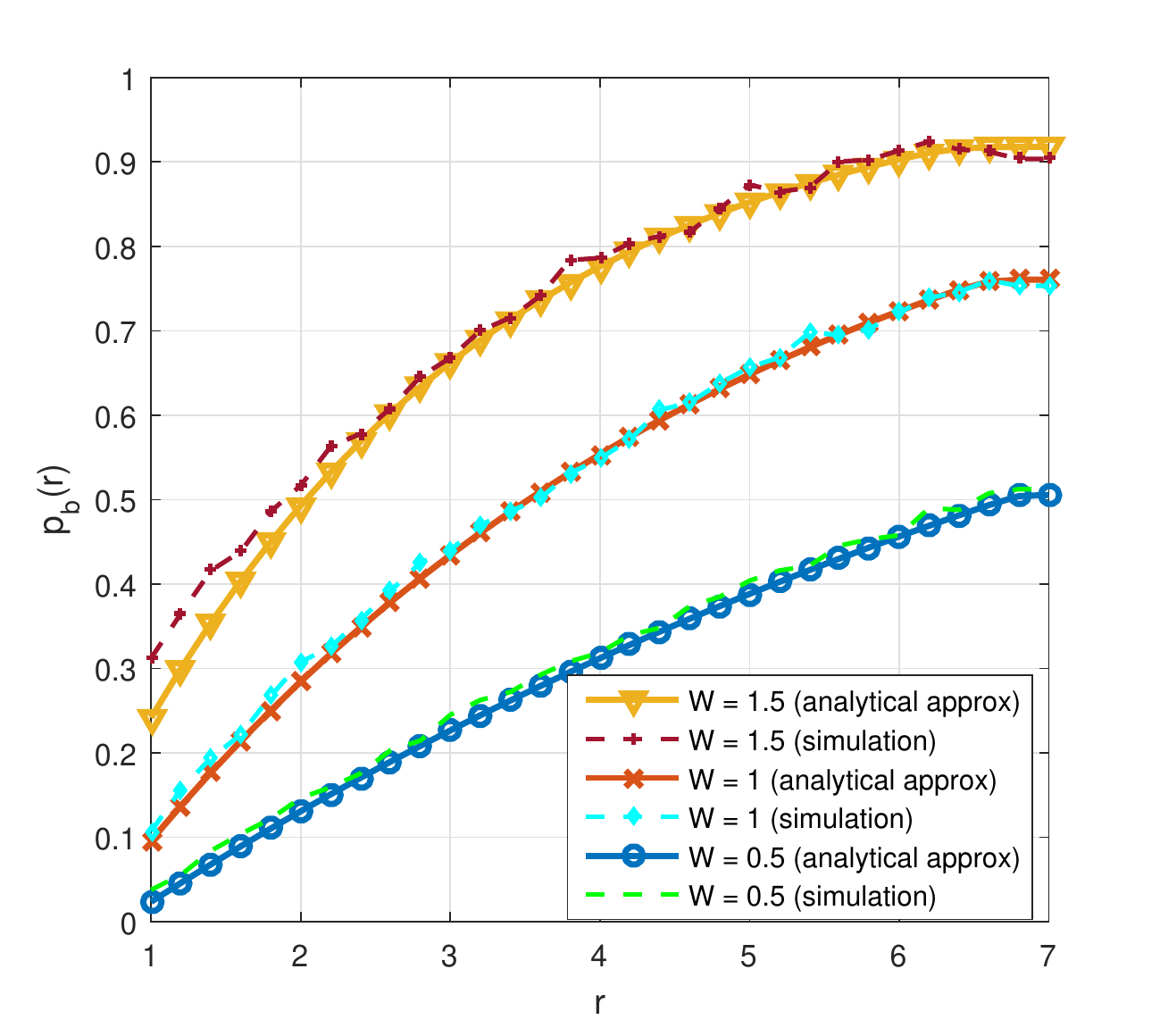}
\label{fig:pbVsW}
}
\caption{The distance-dependent blockage probability as $K$ and $W$ are varied. Here we assume $\rin = 1$ and $\rout = 7$.}
\label{fig:p_br}
\end{figure}

To evaluate $\RB$ under Assumption 4, we compute the mean number of interferers that are not blocked as
\begin{eqnarray}
\rho &=& 2\pi \frac{K}{|\mathcal{A}|}\int_{\rin}^{\rout} \left( 1- p_{\mathsf{b}}(r)\right)r dr.
\end{eqnarray}
Equating this average number of non-blocked interferers to the number of interferers in an equivalent LOS ball of radius $R_\mathsf{B}$ leads to the expression
\begin{eqnarray}
R_\mathsf{B} &=& \left(2\int_{\rin}^{\rout} \left( 1- p_{\mathsf{b}}(r)\right)r dr + {\rin}^2\right)^{0.5}.
\end{eqnarray} Now $\pb(r)$ is approximated by $\tilde{p}_{\mathsf{b}}(r)$, which is a step function with a step up at distance $\RB$,
\be \tilde{p}_{\mathsf{b}}(r)
   =
   \begin{cases}
      0 & \mbox{if ${\rin} \leq r \leq \RB$} \\
      1 & \mbox{if $\RB < r \leq {\rout}$}
   \end{cases}.
\ee   

\subsection{Analysis of Coverage Probability}
\label{ssec:Closed_form}
Under assumptions 1 -- 4, we can derive the probability distribution of $f_{\Omega_i}(w)$ of $\Omega_i$ which are now independent random variables that depend on the location of $X_i$ for $\{ 1,~2,...,~K\}$ as follows \vspace{-0.2in}  
\be
\Omega_i  = \mathfrak{g}_1(\phi_i) R_i^{-\mathfrak{g}_2(R_i)},&& \label{Equation:Omega_LOSBall} 
\ee where
\be
   \mathfrak{g}_1(\phi_i)
   =
   \begin{cases}
      a_i & \mbox{ if $|\phi_i - \phi_0| \leq  \frac{\thetar^{(\mathsf{a})}}{2}$} \\
      b_i & \mbox{ otherwise}
   \end{cases},
   \ee
   and
   \be
   \mathfrak{g}_2(R_i)
  =
   \begin{cases}
      \alphaL & \mbox{ if $R_i \leq R_\mathsf{B}$} \\
      \alphaN & \mbox{ if $R_i > R_\mathsf{B}$}
   \end{cases},     
\ee $a_i = \frac{P_i}{P_0} \Gr$ and $b_i = \frac{P_i}{P_0} \gr$. Additionally, we have
\be \label{Equation:mi}
m_i
   &=&
   \begin{cases}
      \mL & \mbox{ if $R_i \leq R_\mathsf{B}$} \\
      \mN & \mbox{ if $R_i > R_\mathsf{B}$}
   \end{cases}.    
\end{eqnarray}
For $X_i$ drawn from a BPP, $\phi_i$ is uniform random variable in the interval $\left[0, 2\pi\right)$ and $R_i$ has pdf 
\be
\label{Equation:pdf_Ri}
f_{R_i}(r) =& \frac{2\pi r}{|\mathcal{A}|} & \rin \leq r \leq \rout. 
\ee
Next, note that for $|\phi_i - \phi_0| \leq  \frac{\thetar^{(\mathsf{a})}}{2}$ and $\rin \leq R_i \leq R_\mathsf{B}$, $\Omega_i = a_iR_i^{-\alphaL}$ has conditional pdf
\be
\label{Equation:pdf_Omeg_cond}
f_{\Omega_i} ( \omega ) &=& \frac{2\pi\omega^{-\frac{2+\alphaL}{\alphaL}}}{\alphaL|\mathcal{A}|\left( \frac{\pi(\RB^2-\rin^2)}{|\mathcal{A}|}\right)}a_i^{{2}/{\alphaL}}~\mbox{for~ $\frac{a_i}{R_{\mathsf{B}}^{\alphaL}} \leq w \leq \frac{a_i}{\rin^{\alphaL}}$}.
\ee When $|\phi_i - \phi_0| > \frac{\thetar^{(\mathsf{a})}}{2}$, $a_i$ in \eqref{Equation:pdf_Omeg_cond} is replaced with $b_i$, while the $\alphaL$ in \eqref{Equation:pdf_Omeg_cond} is replaced with $\alphaN$ when $\RB < R_i \leq \rout$.
These four cases can be captured by defining a function $\mathcal{D}\left(\omega;[\omega_1, \omega_2];c;\alpha\right)$ as
\be
\mathcal{D}\left(\omega;[\omega_1, \omega_2];c;\alpha\right) &=& \frac{2\pi\omega^{-\frac{2+\alpha}{\alpha}}}{\alpha|\mathcal{A}|}c^{{2}/{\alpha}}\left\lbrace u(\omega - \omega_1)- u(\omega - \omega_2)\right\rbrace, \label{Equation:short_defn}
\ee $u(\cdot)$ being the unit step function, it follows therefore that the pdf of $\Omega_i$ has the form
\be
\nonumber f_{\Omega_i}(w) &=& \frac{\thetar^{(\mathsf{a})}}{2\pi} \left[\mathcal{D}\left(\omega ;\left[\frac{a_i}{R_{\mathsf{B}}^{\alphaL}}, \frac{a_i}{\rin^{\alphaL}}\right];a_i;\alphaL\right) + \mathcal{D}\left(\omega ;\left[\frac{a_i}{\rout^{\alphaN}}, \frac{a_i}{R_{\mathsf{B}}^{\alphaN}}\right];a_i;\alphaN\right)\right] \\
&& \hspace{-0.4in} + \left(1- \frac{\thetar^{(\mathsf{a})}}{2\pi}\right) \left[ \mathcal{D}\left(\omega ;\left[\frac{b_i}{R_{\mathsf{B}}^{\alphaL}}, \frac{b_i}{\rin^{\alphaL}}\right];b_i;\alphaL\right) + \mathcal{D}\left(\omega ;\left[\frac{b_i}{\rout^{\alphaN}}, \frac{b_i}{R_{\mathsf{B}}^{\alphaN}}\right];b_i;\alphaN\right) \right]. \label{Equation:omega_pdf}
\ee
Using the definition in \eqref{Equation:Coverage} and the expression in \eqref{Equation:prod_form}, we can write the spatially averaged CCDF of the SINR by taking an expectation with respect to $\{\Omega_i\}$ as
\be
P_{\mathsf{c}}(\beta) =& \hspace{-0.1in} \mathbb{E}_{\bf{\Omega}}[P_{\mathsf{c}}(\beta, \bf{\Omega})]=& \hspace{-0.1in} e^{-\beta_0\sigma^2}\sum_{\ell = 0}^{m_0 - 1}\frac{(\beta_0\sigma^2)^\ell}{\ell!}\sum_{t = 0}^{\ell}\binom {\ell}{t} \frac{t!}{\sigma^{2t}} \sum_{  {\mathcal S}_t }\left(\prod_{i=1}^{K} \mathbb{E}_{\Omega_i}\left[{\mathcal{G}}_{t_i}(\Omega_i)\right]\right). \label{Equation:CCDF_spa_avg}
\ee
To evaluate $\mathbb{E}_{\Omega_i}\left[{\mathcal{G}}_{t_i}(\Omega_i)\right]$, we note that the integral
\be
\nonumber \int_0^{\infty} \mathcal{D}\left(\omega ;[\omega_1, \omega_2];c;\alpha\right) {\mathcal{G}}_{t_i}(\Omega_i)\mathsf{d}\omega &=& \int_{\omega_1}^{\omega_2} \frac{2\pi\omega^{-\frac{2+\alpha}{\alpha}}}{\alpha |\mathcal{A}|}c^{{2}/{\alpha}} {\mathcal{G}}_{t_i}(\Omega_i)\mathsf{d}\omega \\
\nonumber = (1-\pt)\frac{c^{2/\alpha}}{|\mathcal{A}|}\left[\frac{1}{{\omega_1}^{2/\alpha}} - \frac{1}{{\omega_2}^{2/\alpha}}\right]\hspace{-0.05in} \delta[t_i] \hspace{-0.1in} &+& \hspace{-0.1in}\pt \mathcal{K}_{t_i}(\alpha,c)\hspace{-0.05in} \left\lbrace \pM\hspace{-0.05in}\left[\frac{\mathcal{M}_{t_i}\left(\Gt \omega_1;\alpha\right)}{{\omega_1}^{2/\alpha}} - \frac{\mathcal{M}_{t_i}\left(\Gt \omega_2;\alpha\right)}{{\omega_2}^{2/\alpha}}\right]\right. \\
&+& \hspace{-0.1in} \left.\left(1- \pM\right)\left[\frac{\mathcal{M}_{t_i}\left(\gt \omega_1;\alpha\right)}{{\omega_1}^{2/\alpha}} - \frac{\mathcal{M}_{t_i}\left(\gt \omega_2;\alpha\right)}{{\omega_2}^{2/\alpha}}\right]\right\rbrace , \label{Equation:spatial_avg}
\ee where 
\be
\mathcal{K}_{t_i}(\alpha,c) &=& \frac{2\pi m_i^{m_i}\Gamma(m_i+t_i) \beta_0^{-(m_i + t_i)} c^{2/\alpha}}{\Gamma(m_i)|\mathcal{A}|(t_i!)\alpha}, \\
\mathcal{M}_{t_i}(x; \alpha) &=& \frac{{}_2F_1\left(m_i + t_i, m_i + \frac{2}{\alpha}; m_i + \frac{2}{\alpha} + 1; -\frac{m_i}{x\beta_0}\right)}{x^{m_i}\left(m_i + \frac{2}{\alpha}\right)},
\ee $m_i$ is as given by \eqref{Equation:mi} and ${}_2F_1\left(a, b;c;z\right)$ is the Gauss hypergeometric function.
Using the formulation in \eqref{Equation:spatial_avg} for every term in \eqref{Equation:omega_pdf}, $\mathbb{E}_{\Omega_i}\left[{\mathcal{G}}_{t_i}(\Omega_i)\right]$ can be evaluated so that a closed-form expression for the spatially averaged CCDF of the SINR can be computed from \eqref{Equation:CCDF_spa_avg}. Solving \eqref{Equation:ergSpecEffeciency} through numerical integration, but with \eqref{Equation:CCDF_spa_avg} in the integrand, we can get the spatially averaged ergodic spectral efficiency. 

\section{Results for Random Geometry}
\label{sec:sim_results}
This section gives simulation and numerical results for coverage probability and spectral efficiency, which confirm the validity of assumptions made for spatial averaging. Results generated using assumptions 1 to 3 are all done using a Monte Carlo simulation (which operates by randomly placing the interferers and blockages according to the spatial model, but then computing the conditional outage using \eqref{Equation:prod_form}). Results generated under assumption 4 can be generated using either a simulation or the analytical expression, and the methodology used will be clarified when the result is presented.
The antenna parameters are assumed to be the same as that used earlier, as summarized in Table \ref{table:antenna_param}. The network region $\mathcal{A}$ considered here is an annulus with inner radius $\rin = 0.3$ m and outer radius $\rout = 2.1$ m. The users are assumed to be randomly distributed in $\mathcal{A}$ according to a BPP. The simulation parameters used are summarized in Table \ref{table:num_param}. These are the values used if not otherwise stated. The quantities $K$, $W$ and $\sigma^2$ are parameters we vary for comparison later on. Varying $K$ is equivalent to changing the interferer density $\lambda$ since $|\mathcal{A}|$ is assumed to be fixed. Similarly, varying $W$ amounts to changing the parameters for blockages. Finally, increasing  $\sigma^2$ captures the effect of more noise in the receiver or a lower transmit power.

%\begin{table}
%\caption{Default parameters used for simulation for random geometry}\label{table:sim_param}
%\begin{center}
%\begin{tabular}{|c|c|}
%\hline 
%Parameter & Value \\ 
%\hline
%$d_0$ & 1 \\
%\hline
%$\phi_0$ & 0 \\
%\hline 
%$\rin$ & 1 \\ 
%\hline 
%$\rout$ & 7 \\ 
%\hline 
%$\mL$ & 4 \\ 
%\hline 
%$\mN$ & 2 \\ 
%\hline
%$\alphaL$ & 2 \\ 
%\hline 
%$\alphaN$ & 4 \\ 
%\hline  
%W & 1 \\ 
%\hline 
%$\sigma^2$ & -20 dB \\ 
%\hline 
%$K$ & 36 \\ 
%\hline 
%\end{tabular} 
%\end{center}
%\end{table}

To understand the significance of scaling of the size of the antenna arrays, we plot coverage probability against SINR for 3 cases that have the same product of $\Nt \times \Nr$. This is shown in Fig. \ref{fig:coverage_NtXNr} where we let $\pt = 0.5$ under only Assumption 1. As observed for the fixed geometry case in Section \ref{sec:Num_results}, we see that using more transmit antennas is better than having more receive antennas. This is also seen in Table \ref{table:avgrate_random} which summarizes the ergodic spectral efficiency for various antenna array configurations. The reason for the asymmetrical behavior with respect to $\Nt$ and $\Nr$ is that while larger $\Nt$ results in reduced probability $\pM \propto \frac{1}{\Nt}$ that interferers radiate with larger power $\Gt$, increasing $\Nr$ results in a decreased fraction of interferers falling within the receiver main-lobe which is proportional to ${\thetar^{(\mathsf{a})}} \propto \frac{1}{\sqrt{\Nr}}$. 
\begin{figure}
\centering
\includegraphics[totalheight=3.2in, width=3.2in]{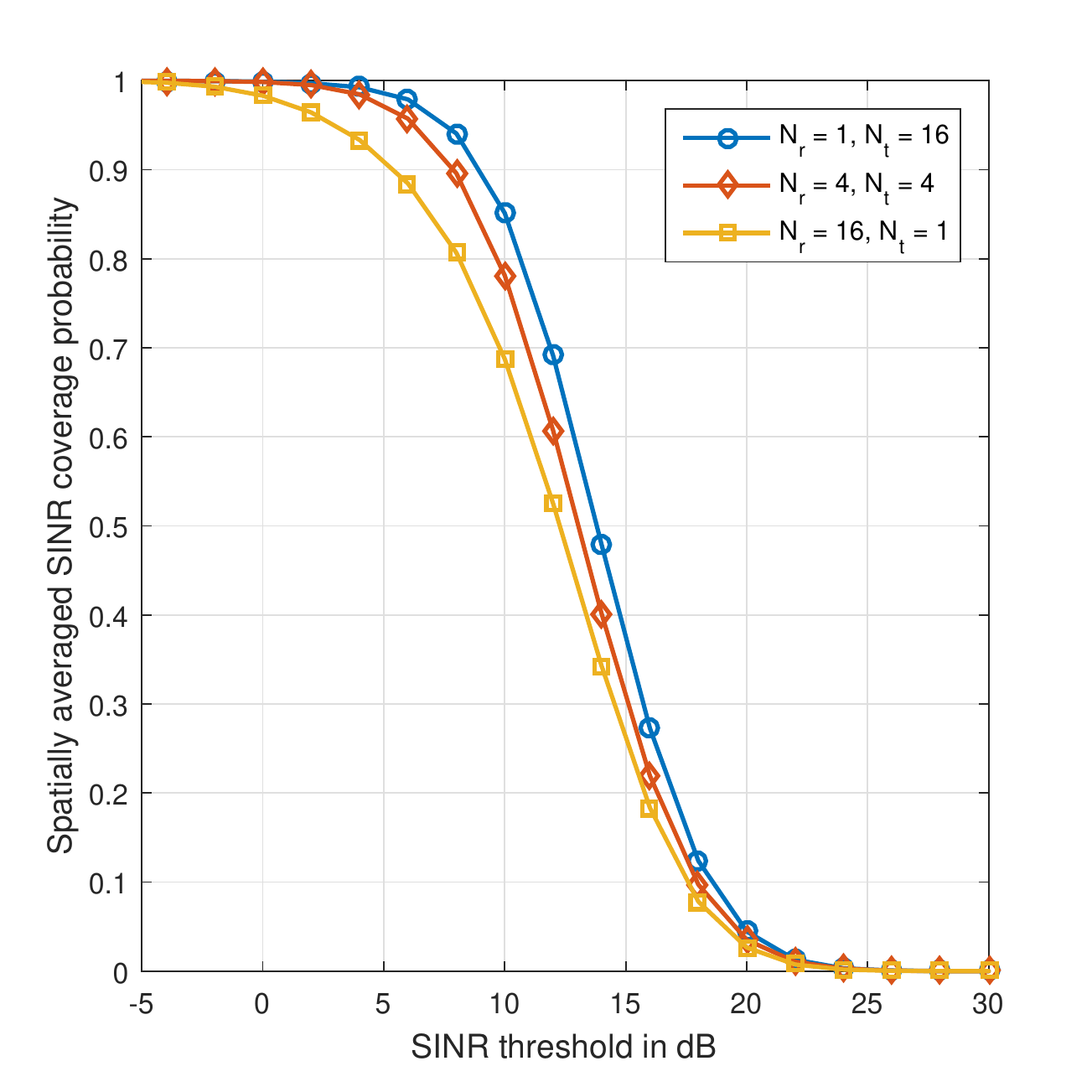}
\vspace{-0.2in}  
\caption{Spatially averaged SINR coverage probability obtained via simulation for three different antenna configurations - $4\times4$, $16\times1$, $1\times16$ and assumption 1 with $\pt = 0.5$. Larger $\Nt$ is advantageous and the performance is not symmetric with respect to $\Nt$ and $\Nr$.}     
\label{fig:coverage_NtXNr} 
\vspace{-0.2in}         
\end{figure}

\begin{table}
\centering
\caption{Spatially averaged ergodic spectral efficiency for various antenna configurations}
\label{table:avgrate_random}
\vspace{-0.1in}
\begin{tabular}{|l||*{3}{c|}}\hline
\backslashbox{$\Nt $}{$\Nr $} & \makebox[2em]{1} & \makebox[2em]{4} & \makebox[2em]{16} \\ \hline
\hline
1 & 0.6465 & 1.7459 & 3.2844 \\
\hline 
4 & 2.0526 & 3.5963 & 5.3523 \\
\hline 
16 & 3.8697 & 5.5886 & 7.4071 \\
\hline 
\end{tabular}
\vspace{-0.2in}
\end{table}
%\begin{table}
%\caption{Spatially averaged rates for various antenna configurations } \label{table:avgrate_random}
%\begin{center}
%\begin{tabular}{|c|c|c|c|}
%\hline 
%{$\Nt$} & {$\Nr$} & Ergodic spectral efficiency (bits/s/Hz)\\ 
%\hline 
%1 & 1 & 0.6926 \\ 
%\hline 
%1 & 4 & 1.0835 \\ 
%\hline
%1 & 16 & 2.4408\\ 
%\hline
%4 & 1 & 0.9707\\ 
%\hline
%4 & 4 &  1.4956\\ 
%\hline
%4 & 16 & 3.3966\\
%\hline
%16 & 1 & 2.4598\\
%\hline
%16 & 4 & 3.2992\\
%\hline
%16 & 16 & 5.3092\\
%\hline
%\end{tabular} 
%\end{center}
%\end{table}

We validate Assumptions 2 -- 4 in Fig. \ref{fig:coverage_Assume}. The plots for the CCDF of spectral efficiency  with and without the assumptions are shown for $\Nt = \Nr = 4$ and $\Nt = \Nr = 16$ with $\pt = 1$. We observe that, though the location of the blockages and the users are dependent in reality (as described by the \textit{orbital model}), the assumptions of treating the blockages and users as two independent BPPs (Assumption 2), associating a distance dependent blockage probability $p_{\mathsf{b}}(r)$ (Assumption 3) and defining the LOS ball (Assumption 4) are all reasonable. The plots of spectral efficiency for each of assumptions 1-4 when $\pt = 1$ are shown in Fig. \ref{fig:coverage_AssumeRate}.
\begin{figure}
\centering
\includegraphics[totalheight=3.2in, width=3.5in]{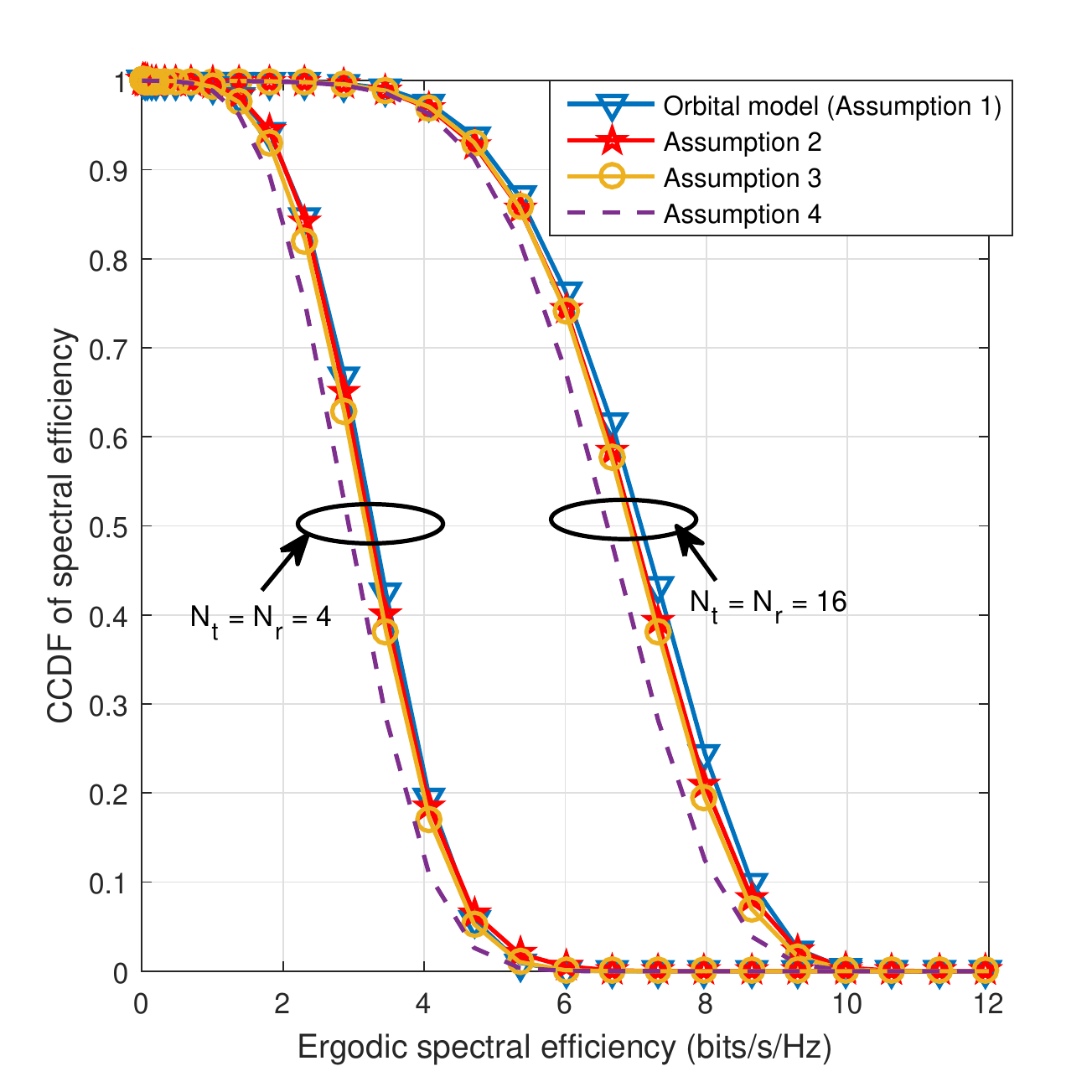}
\vspace{-0.2in}
\caption{CCDF of spatially averaged ergodic spectral efficiency obtained via simulation for various transmitter and receiver antenna configurations $\Nt \times \Nr$ with $\pt = 1$.}
\label{fig:coverage_Assume}
\vspace{-0.2in}        
\end{figure}

\begin{figure}
\centering
\includegraphics[totalheight=3.2in, width=3.5in]{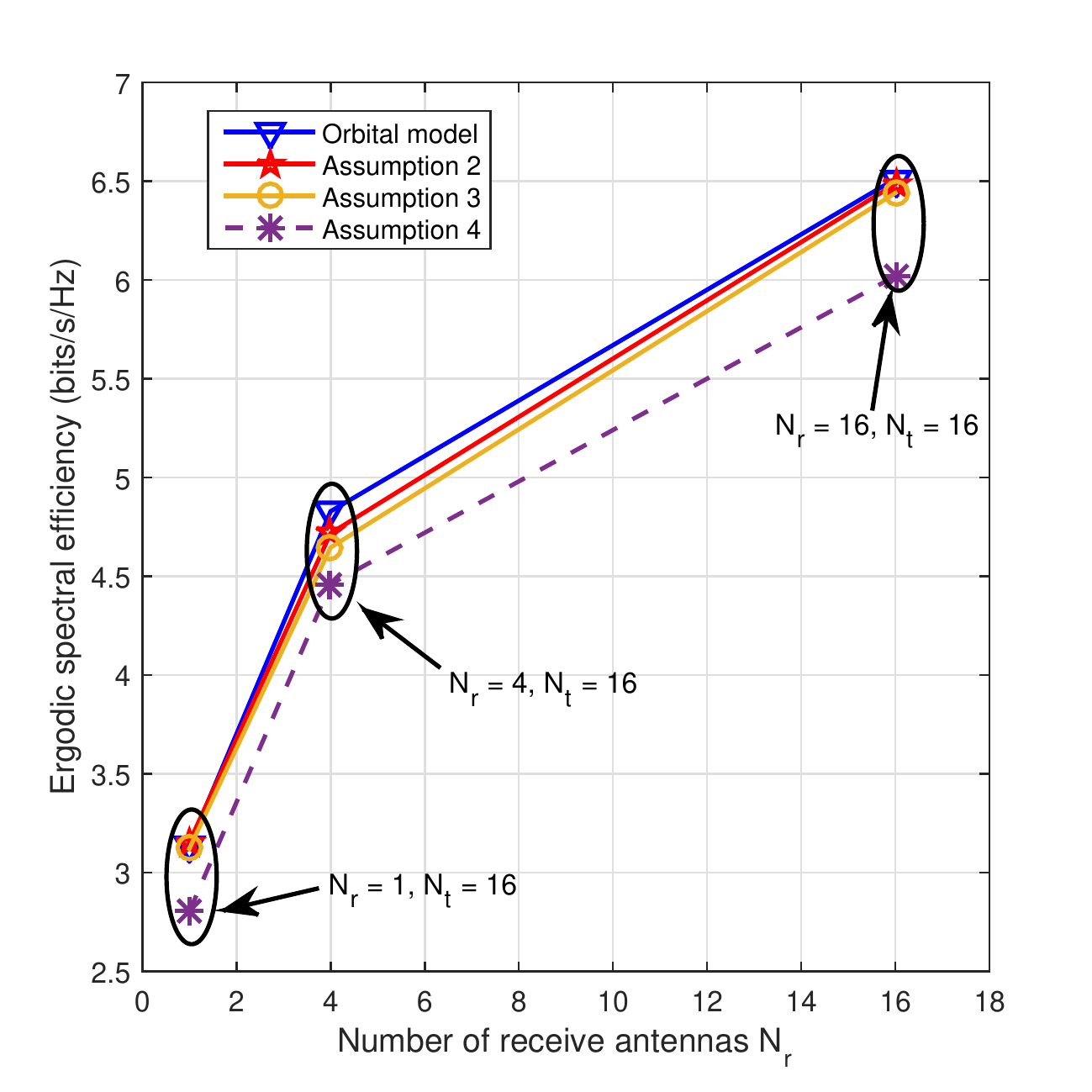}
\vspace{-0.2in}
\caption{Spatially averaged ergodic spectral efficiency from simulation when $\pt = 1$ for various receiver antenna configurations and $\Nt = 16$ with and without Assumptions 2 -- 4 in Section \ref{sec:Assumptions}.} 
\label{fig:coverage_AssumeRate}
\vspace{-0.2in}        
\end{figure}

The plots in Fig. \ref{fig:Analytic_fig} show the CCDF of the SINR obtained using the analytic expressions derived in Section \ref{ssec:Closed_form} and compares it with the actual simulation results under Assumptions 1 and 4. It is seen that the analytic expressions match exactly with the setting under Assumption 4 wherein we consider all the interferers within the LOS ball as unblocked and those outside as blocked from the reference receiver.

\begin{figure}
\centering
\includegraphics[totalheight=3.2in, width=3.5in]{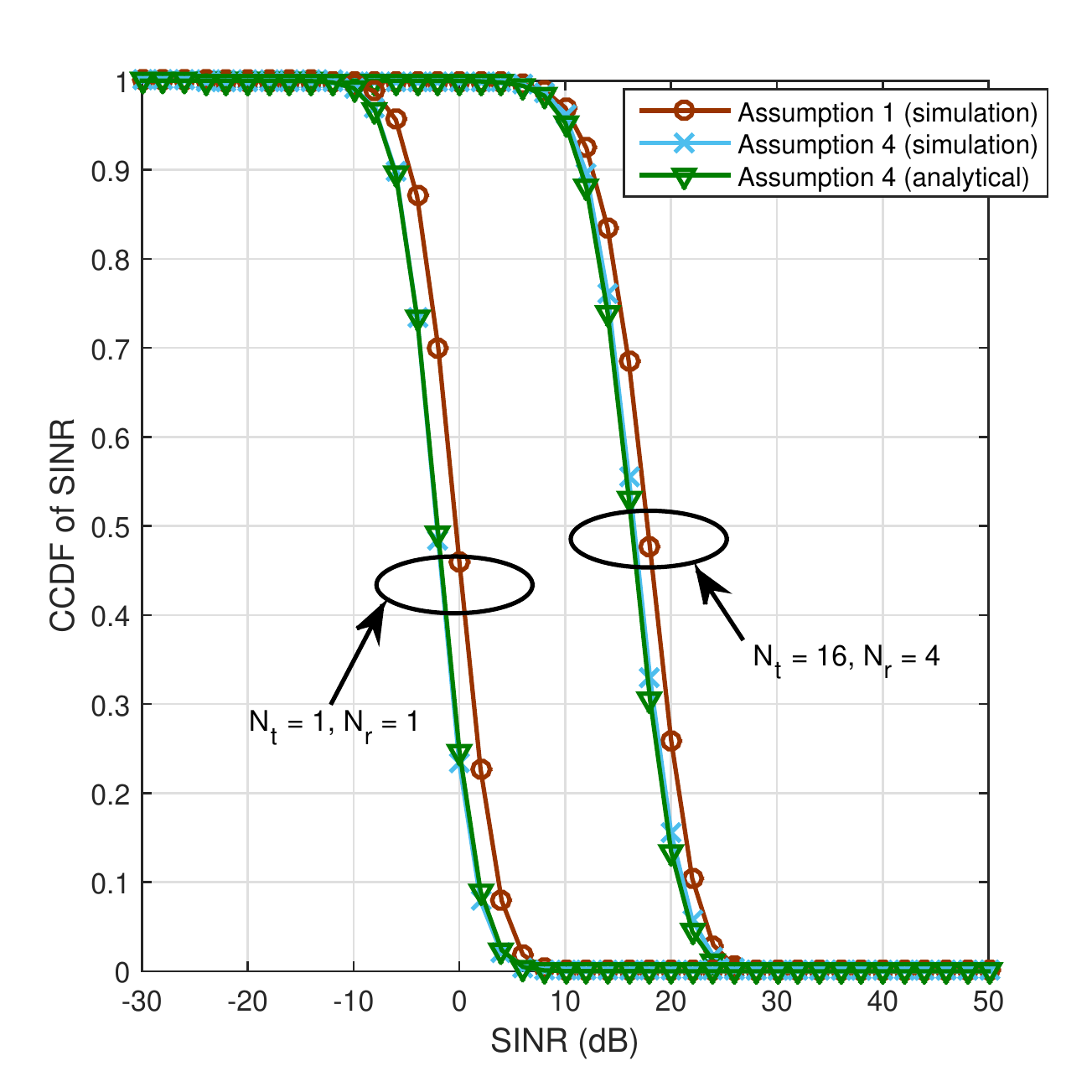}
\vspace{-0.2in}
\caption{Plot showing the CCDF of spatially averaged SINR obtained from simulation and analytic closed-form expressions for different antenna configuration with $\pt = 0.7$.} 
\label{fig:Analytic_fig}
\vspace{-0.2in}        
\end{figure}

Next we look at the dependence of the system performance on $W$, the diameter of the blockages. We define the throughput as the product of $\pt$ and the ergodic spectral efficiency. The plot of throughput versus $W$ in Fig. \ref{fig:TputVsW} shows that as $W$ is increased, the throughput improves. This is because, with larger $W$ and for a fixed interferer density, the interfering signals get more blocked thus improving the SINR. The plots in Fig. \ref{fig:TputVsW} are for $\Nt = \Nr = 4$ and using the analytic expressions derived in Section \ref{ssec:Closed_form}.

\begin{figure}
\centering
\includegraphics[totalheight=3.2in, width=3.5in]{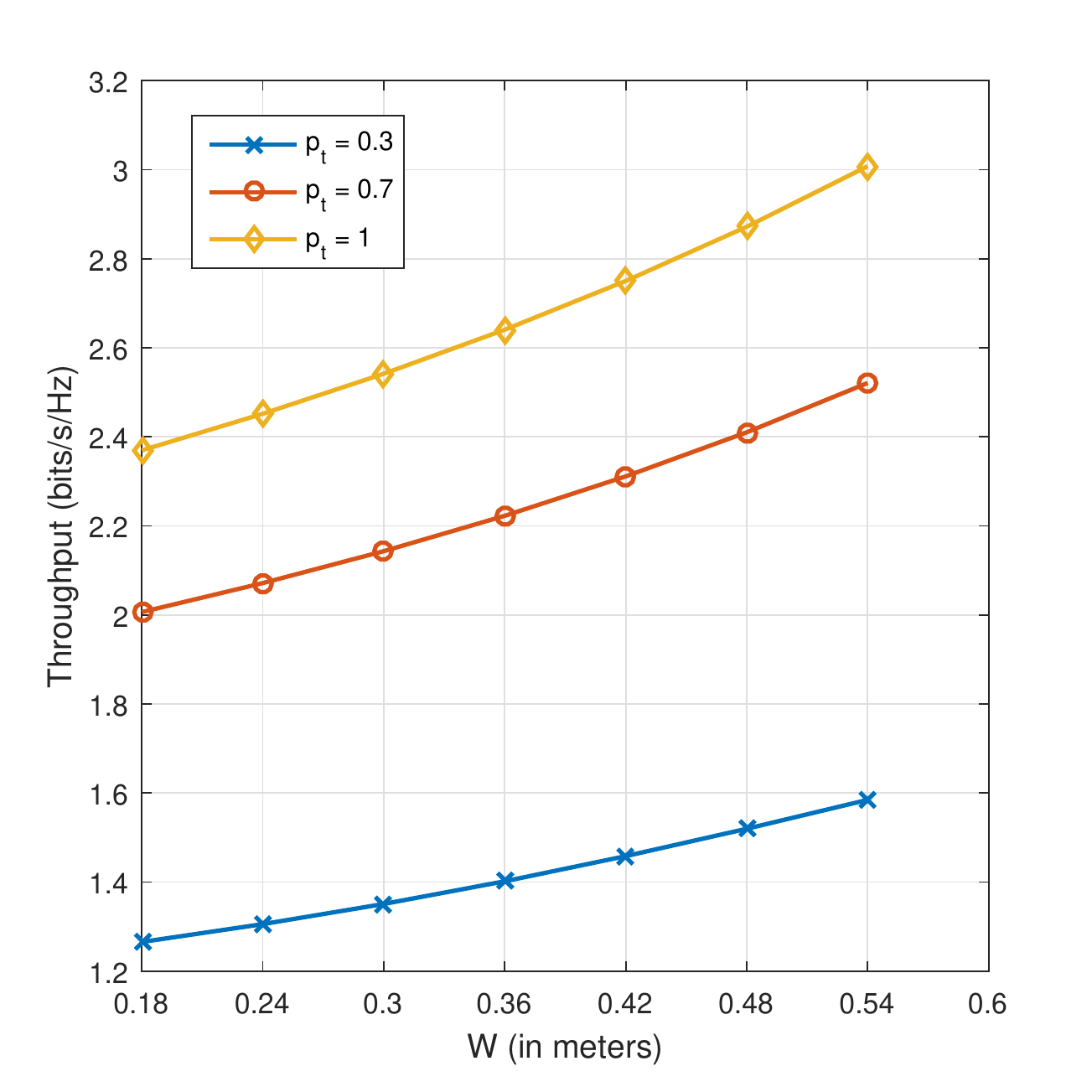}
\vspace{-0.2in}
\caption{Spatially averaged throughput vs. $W$ for different values of random-access probability $\pt$ using the analytic expressions. Larger blockage diameter results in better throughput as the interferers are effectively blocked.} 
\label{fig:TputVsW}
\vspace{-0.2in}
\end{figure}

In Fig. \ref{fig:coverageVsdensity}, the variation of SINR coverage probability is plotted as a function of $\lambda$, the interferer density. We fix $\Nt = 4$, $\Nr = 16$ and $\pt = 1$ for comparison and use the previously derived analytic expressions for the plots. It is seen that as $\lambda$ is increased the SINR coverage decreases rapidly initially. However, with very high density, blocking probability also increases, hence showing lower rate of decrease with increasing $\lambda$ in the plots later on. This also corroborates our assumption that (the users wearing) the interferers are also the source of blockages in the indoor wearables environment.

\begin{figure}
\centering
\includegraphics[totalheight=3.2in, width=3.5in]{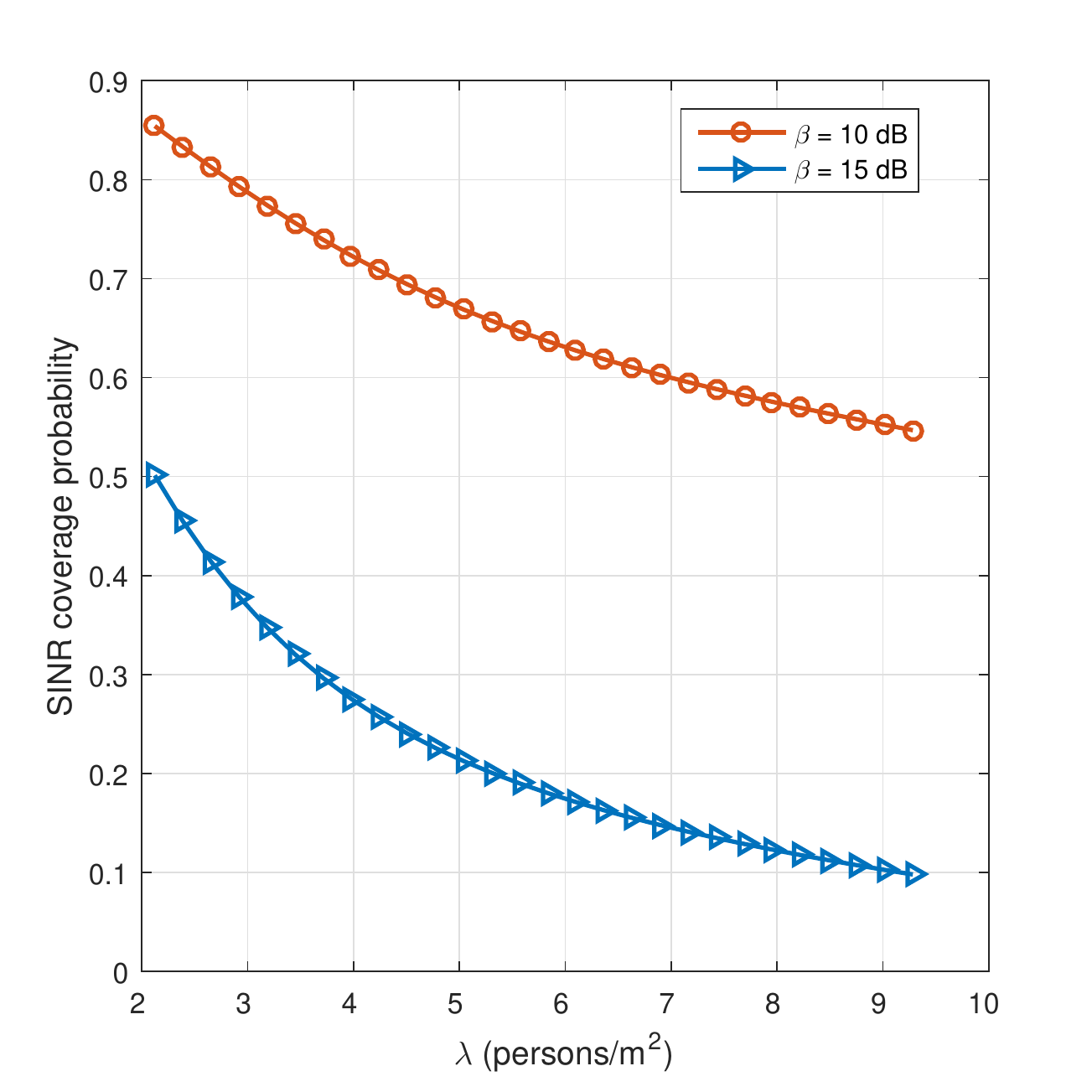}
\vspace{-0.2in}
\caption{Spatially averaged SINR coverage probability vs. $\lambda$ for different values of SINR threshold $\beta$ using the analytic expressions. Here, we let $\Nt = 4$, $\Nr = 16$ and $\pt = 1$.} 
\label{fig:coverageVsdensity}
\vspace{-0.2in}        
\end{figure}

Fig. \ref{fig:coverageVsSNR} shows the variation of ergodic spectral efficiency as we vary $\sigma^2$. Here we let $\Nt = \Nr = 4$ and use the analytic results in Section \ref{ssec:Closed_form}. It is seen that for smaller values of $\sigma^2$, the system is indeed interference limited as changing $\pt$ results in significant change in the system performance. However as $\sigma^2$ is increased, the system becomes noise limited and different values for random-access probabilities of the interferers result in little change in the SINR distribution and hence the ergodic spectral efficiency.
\begin{figure}
\centering
\includegraphics[totalheight=3.2in, width=3.5in]{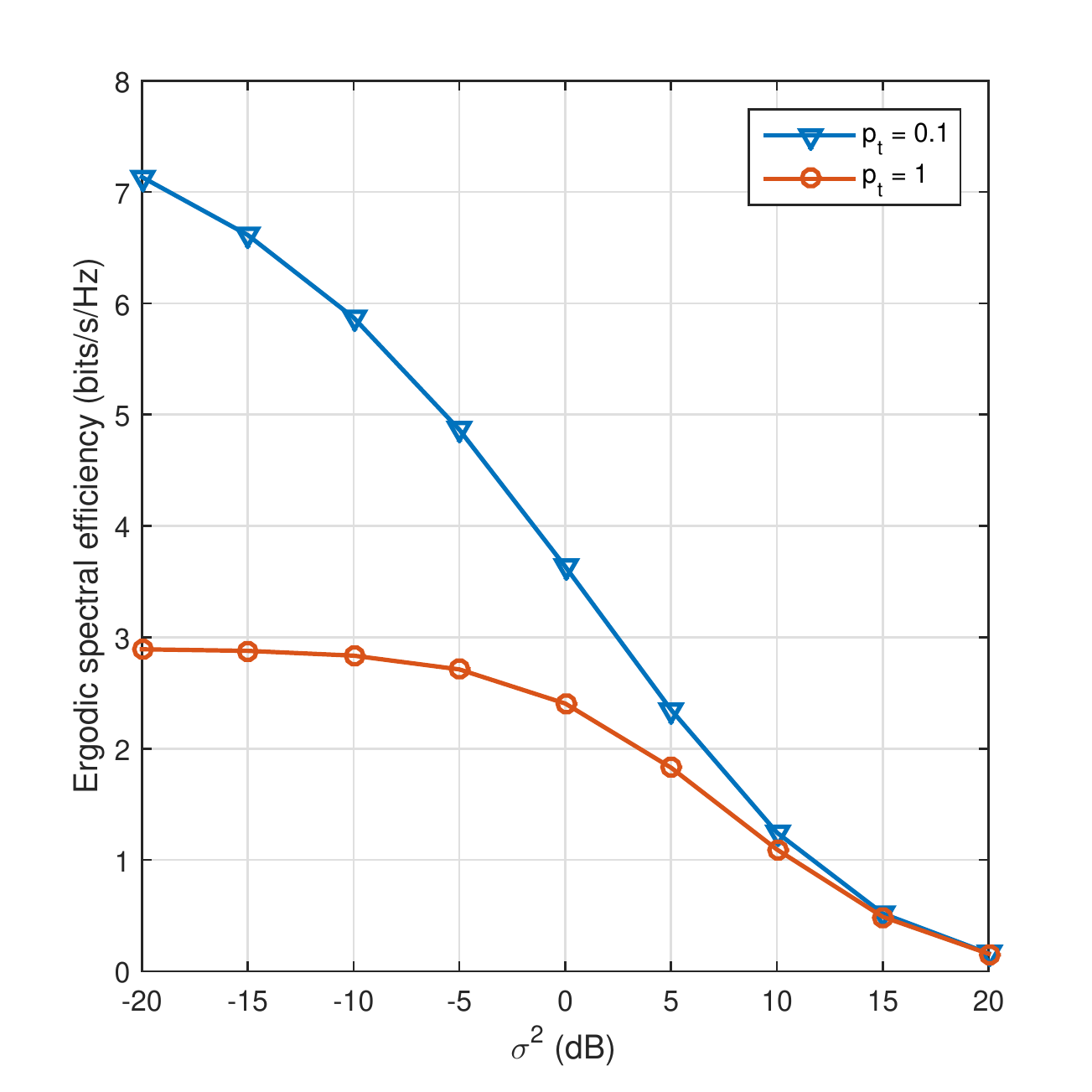}
\vspace{-0.2in}
\caption{Spatially averaged ergodic spectral efficiency vs $\sigma^2$ for two different values of $\pt$ using the analytic expressions when $\Nt = \Nr = 4$.} 
\label{fig:coverageVsSNR}
\vspace{-0.2in}    
\end{figure}

%%%%%%%%%%%%%%%%%%%%%%%%%%%%%%%%%
\section{Conclusion}
\label{sec:Concl}
In this paper, we analyzed the performance of a mmWave wearable communication network operating in a finite region like that inside a train car. To model the sensitivity of mmWave signal propagation to the presence of human bodies in the network, we incorporated different path-loss and small-scale fading parameters depending on whether a link is blocked or not. It was seen that both interference and the probability of blockage of the interference signals are large when the crowd density is high, so that the SINR coverage probability decreases at a much lower rate with higher crowd density. We considered fixed as well as random positions for the interfering transmitters and assessed the impact of antenna parameters such as array gain and beamwidth on coverage and ergodic spectral efficiency of the system. It was seen that antenna main-lobe directivity and array gain play a crucial role in achieving giga-bits per second performance for wearable networks in a crowd. We proposed several assumptions and a model to analyze the system performance when the interferers are located at random locations. These gave closed-form expressions for spatially averaged coverage probability for mmWave wearable communication network when the user is located at the center of a dense crowd and the number of users is finite. The validity of the closed-form analytic results and the assumptions were confirmed against simulations. The analytic modeling presented in this paper serves as a first step towards characterizing SINR performance of mmWave based ad-hoc networks in a finite but crowded environment, and enables one to avoid simulations to predict performance. 

For future work, it would be interesting to consider further refinements to the model including incorporation of 3D locations for the devices, and explicit modeling of reflections of the mmWave signals from the boundaries of the finite network region. The work in this paper can easily be extended to the case that the reference link can be blocked by the user's own body.  The procedure would involve finding two conditional outage probabilities, one conditioned on the link not being blocked by the user (using the procedure outlined in this paper) and the other conditioned on the link being blocked by the user (adapting the procedures so that the reference link's path loss is $\alphaN$ and Nakagami factor is $\mN$).  The two probabilities can then be weighted by the probability of self-blockage, which can be determined based on the assumed spatial models. Using a more refined model to capture this self-blockage and incorporating it in the analysis is an interesting topic for future work.

%%%%%%%%%%%%%%%%%%%%%%%%%%%%%%%%%
\section{Acknowledgement}
The authors would like to thank Salvatore Talarico for his programming assistance and Geordie George for his discussion on antenna gain pattern modeling.

\appendices

\section{PROOF OF LEMMA \ref{lemma:snowboard}}
\label{Appdx:proof_snowboard}
The blockages are drawn from a BPP. Consider a transmitter $X_i$ located at distance $|X_i| = r$ from the reference receiver.  Its signal will be blocked if there is a blockage inside a certain subregion of $\mathcal A$, which we will call the \emph{blocking region} of $X_i$ (i.e., $X_i$ is blocked if there is an object in its blocking region).  Since $\mathcal{A}$ is a circular disk with inner radius $\rin$ and outer radius $\rout$, the blocking region looks like Fig. \ref{fig:blocking_region1} if $\rin \leq r \leq \rout - \frac{W}{2}$ and like Fig. \ref{fig:blocking_region2} if $\rout - \frac{W}{2} < r \leq \rout$. 

\begin{figure}
\centering
\subfigure[Case when $\rin \leq r \leq \rout - \frac{W}{2}$]{
\includegraphics[scale=0.4]{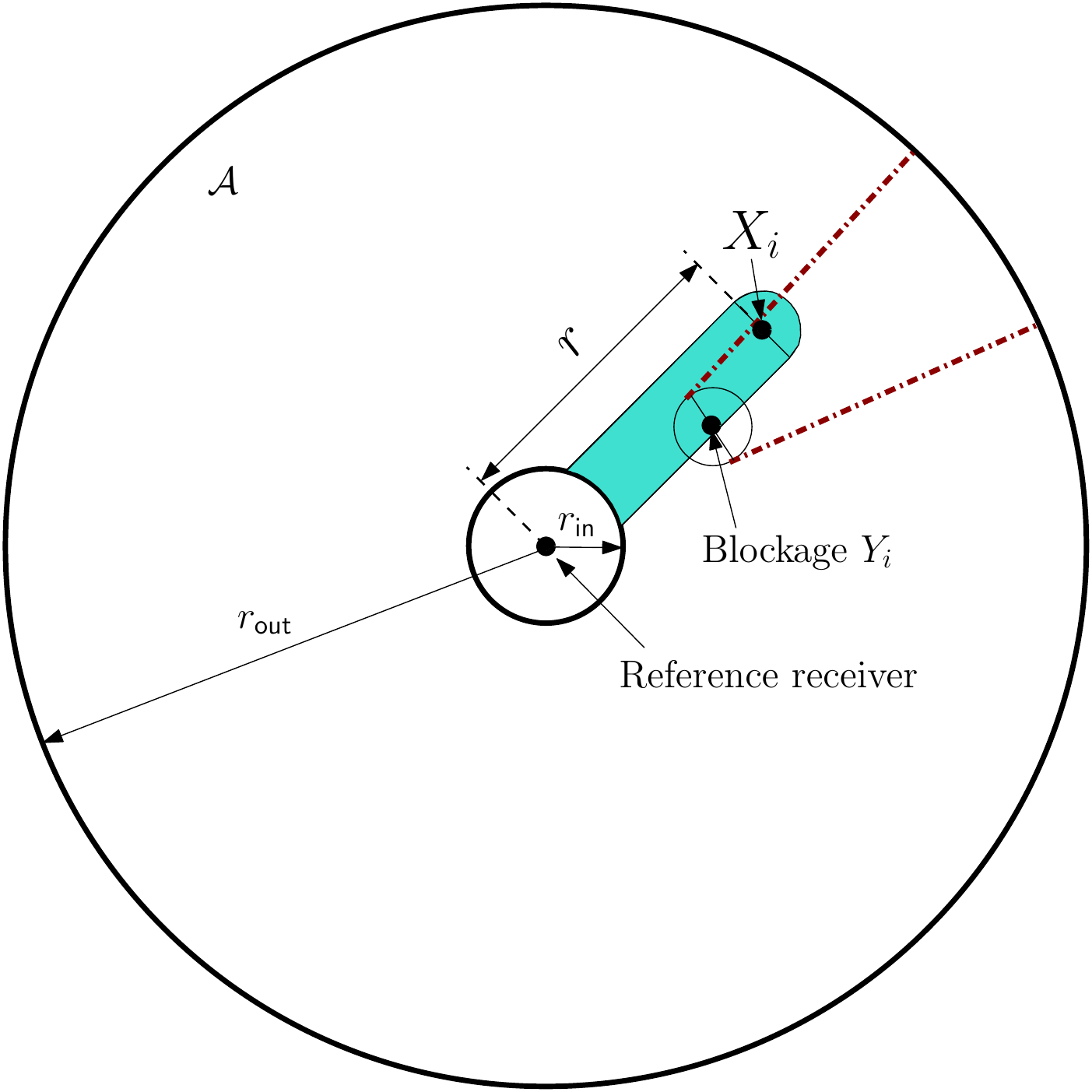}
\label{fig:blocking_region1}}
\hspace{.5in}
\subfigure[Case when $\rout - \frac{W}{2} < r \leq \rout$]{
\includegraphics[scale=.4]{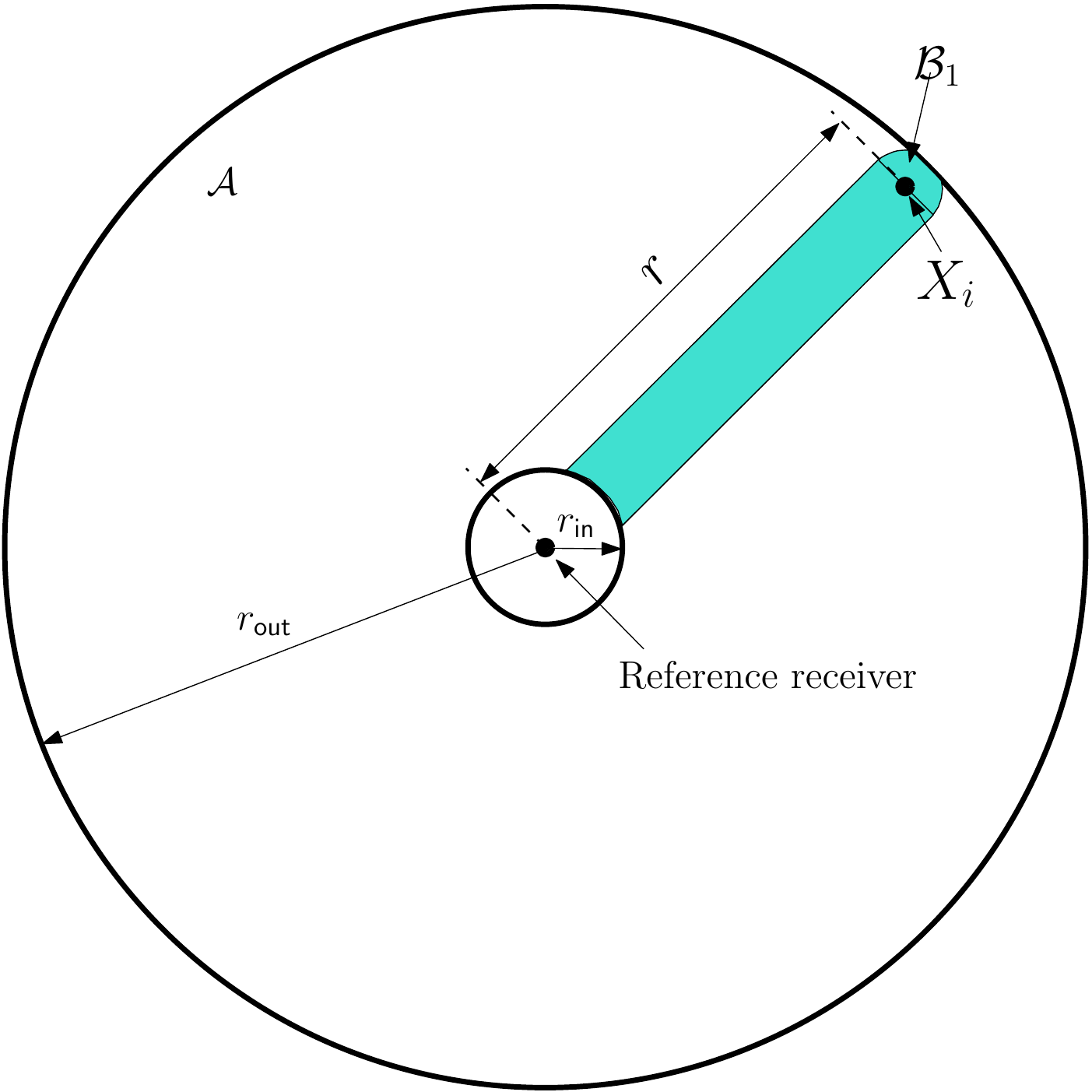}
\label{fig:blocking_region2}}
\caption{Figure showing the blocking region for interferer $X_i$ with $|X_i| = r$ for two different cases. The blocking cone of a blockage $B_i$ of diameter $W$ that lies within the blocking region of $X_i$ is also shown for illustration}     
\label{fig:blocking_zone}        
\end{figure}
Let $p_{\mathsf{b}}(i,j)$ be the probability that blockage $B_j$ blocks transmitter $X_i$ with $|X_i| = r$.  Since the blockages are placed uniformly at random, the probability that the blockage is inside the blocking region is equal to the ratio of the area of the corresponding blocking region and the overall network. For ${\rin} \leq r \leq {\rout} - \frac{W}{2}$, the area of the blocking region can be evaluated as follows. The area can be split into regions as shown in Fig. \ref{fig:app_proof1}, where region $\mathcal{A}_1$ is a sector of the circle with radius $\rin$ and subtended angle $\psi = 2\arcsin \frac{W}{2\rin}$. Region $\mathcal{A}_2$ corresponds to two identical right triangles with base length $W/2$ and height $\sqrt{\rin^2 - \left(\frac{W}{2}\right)^2}$ and region $\mathcal{A}_3$ is a semicircular disk of radius $W/2$. Hence, the area of the shaded region in Fig. \ref{fig:blocking_region1} is $rW + \frac{\pi W^2}{8} - |\mathcal{A}_1| - 2|\mathcal{A}_2| $, where $|{\mathcal{A}}_2| = \frac{W}{4}\sqrt{{\rin}^2 - \left(\frac{W}{2}\right)^2} $ and $|{\mathcal{A}}_1| = {\rin}^2\arcsin\left(\frac{W}{2{\rin}}\right)$.
\begin{figure}
\centering
\includegraphics[height = 2in, width = 2in]{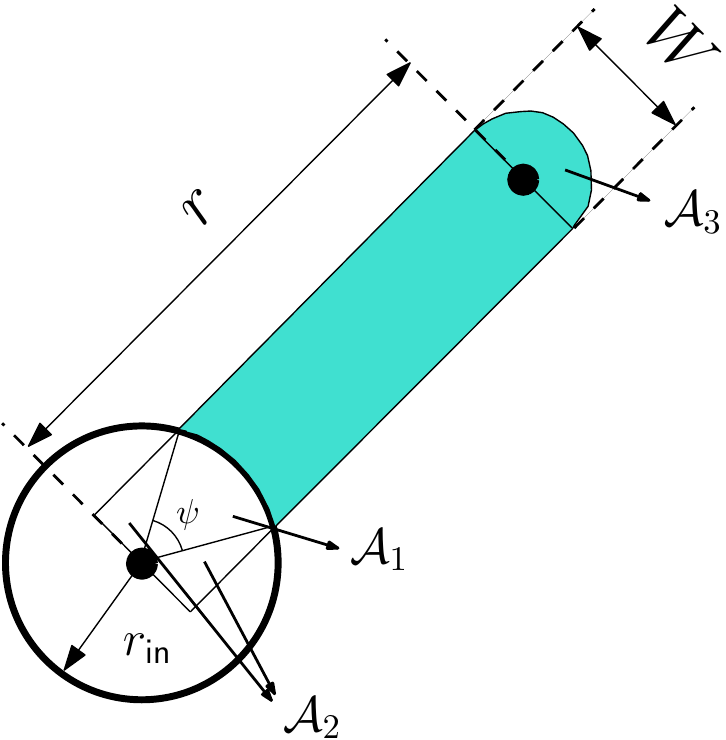}
\caption{Evaluation of area of blocking region for ${\rin} \leq r \leq {\rout} - \frac{W}{2}$}
\label{fig:app_proof1}
\vspace{-0.2in}
\end{figure}
For the two cases shown in Fig. \ref{fig:blocking_zone}, we would then have,
\be
p_{\mathsf{b}}(i,j)
   & = &
   \begin{cases}
        \frac{rW - \frac{W}{2}\sqrt{{\rin}^2 - \left(\frac{W}{2}\right)^2}-{\rin}^2\arcsin\left(\frac{W}{2{\rin}}\right) + \frac{\pi W^2}{8} }{|\mathcal{A}|}   & \mbox{if ${\rin} \leq r \leq {\rout} - \frac{W}{2}$} \\
         \frac{rW - \frac{W}{2}\sqrt{{\rin}^2 - \left(\frac{W}{2}\right)^2}-{\rin}^2\arcsin\left(\frac{W}{2{\rin}}\right) + \nu }{|\mathcal{A}|} & \mbox{if ${\rout} - \frac{W}{2} \leq r \leq {\rout}$} 
   \end{cases}, \label{Equation:Pairwise_block}
\ee where
\be
\nonumber \nu &=& \left(\frac{W}{2}\right)^2\arcsin\left(\frac{{\rout}^2 - \left(\frac{W}{2}\right)^2 - r^2}{rW}\right)+ {\rout}^2\arccos\left(\frac{{\rout}^2 - \left(\frac{W}{2}\right)^2 + r^2}{2r{\rout}}\right)\\
&&~~~~~~~~~~~~~~~~~~~~~-2\sqrt{s(s-r)(s-\frac{W}{2})(s- \frac{\rout}{2})}; s = \frac{{\rout}+r+W/2}{2} \nonumber
\ee is the area of region $\mathcal{B}_1$ indicated in Fig. \ref{fig:blocking_region2}.
Since the blockages are independent, the transmitter will be blocked if there are \emph{any} blockages --- or, equivalently, will not be blocked only if there are no blockages in its blocking region,  Thus, the probability that $X_i$ located at $|X_i| = r$ is blocked is
\be
p_{\mathsf{b}}(r)
& = &
1 - \prod_{j=1}^{K} \left( 1 -  p_{\mathsf{b}}(i,j) \right),
\ee resulting in the form given in Lemma \ref{lemma:snowboard}.

%\balance

\bibliographystyle{ieeetr}

\end{document}